\documentclass[11pt]{article}

\usepackage{fullpage}

\usepackage{amsthm,amsmath,amsfonts,amssymb}  
\usepackage{xspace,enumerate}
\usepackage[utf8]{inputenc}
\usepackage{thmtools}
\usepackage{thm-restate}
\usepackage{authblk}
\usepackage[hypertexnames=false,colorlinks=true,urlcolor=Blue,citecolor=Green,linkcolor=BrickRed]{hyperref}
\usepackage[dvipsnames,usenames]{color}
\usepackage{graphicx}
\usepackage[noadjust]{cite}
\usepackage{microtype}
\usepackage{todonotes}
\usepackage{thmtools}
\usepackage{thm-restate}
\usepackage[capitalise]{cleveref}
\usepackage{algorithm}
\usepackage{algorithmicx}
\usepackage[pagewise]{lineno}
\usepackage{algpseudocode}
\nolinenumbers

\theoremstyle{plain}
\newtheorem{theorem}{Theorem}
\newtheorem{lemma}[theorem]{Lemma} 
\newtheorem{corollary}[theorem]{Corollary}  
\newtheorem{proposition}[theorem]{Proposition}  
\newtheorem{fact}[theorem]{Fact}

\newtheorem{definition}[theorem]{Definition}

\usepackage{todonotes}

\newcommand{\cost}{\mathrm{cost}}
\newcommand{\cut}{\mathrm{cut}}
\newcommand{\cc}{\mathrm{cc}}
\newcommand{\parent}{\mathrm{parent}}
\newcommand{\ith}{i^{\scriptsize \mbox{{\rm th}}}}
\newcommand{\eps}{\varepsilon}
\newcommand{\R}{\mathbb{R}}

\newcommand{\Oh}{\mathcal{O}}
\newcommand{\tOh}{\widetilde{\mathcal O}}
\newcommand{\cC}{\mathcal{C}}
\newcommand{\cL}{\mathcal{L}}
\newcommand{\cS}{\mathcal{S}}
\newcommand{\cT}{\mathcal{T}}
\newcommand{\cA}{\mathcal{A}}
\newcommand{\cB}{\mathcal{B}}
\newcommand{\cD}{\mathcal{D}}
\newcommand{\cF}{\mathcal{F}}

\newcommand{\cK}{\mathcal{K}}
\newcommand{\cM}{\mathcal{M}}

\newcommand{\cP}{\mathcal{P}}
\newcommand{\cW}{\mathcal{W}}
\newcommand{\cX}{\mathcal{X}}
\newcommand{\cY}{\mathcal{Y}}
\newcommand{\cZ}{\mathcal{Z}}
\newcommand{\edgecolor}{\mathrm{color}}
\newcommand{\score}{\mathrm{score}}
\newcommand{\update}{\mathrm{update}}

\newcommand{\interval}{\mathrm{int}}
\newcommand{\NULL}{\mathrm{NULL}}

\newcommand{\Found}{\mathrm{Found}}

\DeclareMathOperator*{\argmin}{arg\,min}
\newcommand{\lca}{\mathrm{lca}}
\newcommand{\Round}{\mathsf{RoundEdges}}

\algnewcommand\True{\textbf{true}\space}
\algnewcommand\False{\textbf{false}\space}
\newcommand{\pluseq}{\mathrel{+}=}
\newcommand{\minuseq}{\mathrel{-}=}
\newcommand{\M}{M^{(h,h')}}

\begin{document}
\title{Finding the KT partition of a weighted graph in near-linear time}
\author{Simon Apers\thanks{IRIF, CNRS, Paris. Email: smgapers@gmail.com} \qquad
Pawe{\l} Gawrychowski\thanks{Institute of Computer Science, University of Wroc{\l}aw, Poland. Email: gawry@cs.uni.wroc.pl. Partially supported by the Bekker programme of the Polish National Agency for Academic Exchange (PPN/BEK/2020/1/00444).} \qquad
Troy Lee\thanks{Centre for Quantum Software and Information, University of Technology Sydney. Email: troyjlee@gmail.com. Supported in part by the Australian Research Council Grant No: DP200100950}}
\date{\vspace{-5mm}}
\maketitle

\begin{abstract}
In a breakthrough work, Kawarabayashi and Thorup (J.~ACM'19) gave a near-linear time deterministic algorithm to compute the weight of a minimum cut in 
a simple graph $G = (V,E)$.  A key component of this algorithm is finding the $(1+\eps)$-KT partition of $G$, the coarsest partition $\{P_1, \ldots, P_k\}$ of $V$ such that for every non-trivial $(1+\eps)$-near minimum 
cut with sides $\{S, \bar{S}\}$ it holds that $P_i$ is contained in either $S$ or $\bar{S}$, for $i=1, \ldots, k$.  In this work we give a near-linear time randomized algorithm to find the $(1+\eps)$-KT partition of a 
\emph{weighted} graph.  Our algorithm is quite different
from that of Kawarabayashi and Thorup and builds on Karger's framework of tree-respecting cuts (J.~ACM'00).

We describe a number of applications of the algorithm.
(i) The algorithm makes progress towards a more efficient algorithm for constructing the polygon representation of the set of near-minimum cuts in a graph.
This is a generalization of the cactus representation, and was initially described by Bencz\'ur (FOCS'95).  (ii) We improve the time complexity of a recent quantum algorithm for minimum cut in a simple graph in the adjacency list model from 
$\tOh(n^{3/2})$ to $\tOh(\sqrt{mn})$, when the graph has  $n$ vertices and $m$ edges. (iii) We describe a new type of randomized 
algorithm for minimum cut in simple graphs with complexity $\Oh(m + n \log^{6} n)$.
For graphs that are not too sparse, this matches the complexity of the current best $\Oh(m + n \log^2 n)$ algorithm which uses a different approach based on random contractions.

The key technical contribution of our work is the following.  Given a weighted graph $G$ with $m$ edges and a spanning tree $T$ of $G$, consider the graph $H$ whose nodes are the edges of $T$, and 
where there is an edge between two nodes of $H$ iff the corresponding 2-respecting cut of $T$ is a non-trivial near-minimum cut of $G$.  We give a $\Oh(m \log^4 n)$ time deterministic algorithm to compute a 
spanning forest of $H$.
\end{abstract}

\section{Introduction}

Given a weighted and undirected graph $G$ with $n$ vertices and $m$ edges,\footnote{Throughout this paper we will use $n$ and $m$ to denote the number of vertices and edges of the input graph.}
the minimum cut problem is to find the minimum weight $\lambda(G)$ of a set of edges whose removal disconnects $G$.  When $G$ is unweighted, this is simply the minimum number of edges 
whose removal disconnects $G$, also known as the edge connectivity of $G$.  The minimum cut problem is a fundamental problem in theoretical computer science whose study goes back to at least the 1960s when the first polynomial time algorithm 
computing edge connectivity was given by Gomory and Hu \cite{GH61}.  In the current state-of-the-art, there are near-linear time randomized algorithms for the minimum cut problem in weighted graphs \cite{Karger, 
GawrychowskiMW20, MN20} and near-linear time \emph{deterministic} algorithms in the case of simple graphs\footnote{A simple graph is an unweighted graph with no self loops and at most one edge between any pair of vertices.}
\cite{KT19, HRW20}.  Very recently, Li \cite{Li21} has given an almost-linear time (i.e.\ time $\Oh(m^{1+ o(1)})$) deterministic algorithm for weighted graphs as well.

The best known algorithms for weighted graphs all rely on a framework developed by Karger \cite{Karger} which, for an input graph $G$, relies on finding $\Oh(\log n)$ spanning trees of $G$ such that with high 
probability one of these spanning trees will contain at most 2 edges from a minimum cut of $G$.  In this case the cut is said to 2-respect the tree.  A key insight of Karger is that, given a spanning 
tree $T$ of $G$, the problem of finding a 2-respecting cut of $T$ that has minimum weight in $G$ can be solved deterministically in near-linear time, specifically time $\Oh(m \log^2 n)$.  
After standing for 20 years, the bound for this minimum-weight 2-respecting cut problem was recently improved by Gawrychowski, Mozes, and 
Weimann \cite{GawrychowskiMW20}, who gave a deterministic $\Oh(m\log n)$ time algorithm, and independently by Mukhopadhyay and Nanongkai \cite{MN20} who gave a randomized algorithm with 
complexity $\Oh(m\log n + n \log^4 n)$.

The best algorithms in the case of a simple graph $G$ rely on a quite 
different approach, pioneered by Kawarabayashi and Thorup \cite{KT19}.  This approach begins by finding the minimum degree $d$ of a vertex in $G$.  Then 
the question becomes if there is a non-trivial cut, i.e.\ a cut where both sides of the corresponding bipartition have cardinality at least 2, whose weight is less than $d$.  This problem 
is solved by finding what we call the $(1+\eps)$-\emph{KT partition} of the graph.  Let $\cB_\eps^{nt}(G)$ be the set of all bipartitions $\{S, \bar{S}\}$ of the vertex set corresponding to non-trivial cuts whose weight is at most $(1+\eps)\lambda(G)$.  
The $(1+\eps)$-KT partition of $G$ is the coarsest partition $\{P_1, \ldots, P_k\}$ of the vertex set such that for any $\{S,\bar{S}\} \in \cB_\eps^{nt}(G)$ it holds that $P_i$ is contained in either $S$ or $\bar{S}$, for each $i=1, \ldots, k$.  
If one considers the multigraph $G'$ formed from $G$ by identifying vertices in the same set $P_i$, then $G'$ preserves all non-trivial 
$(1+\eps)$-near minimum cuts of $G$.  Kawarabayashi and Thorup further show that for any $\eps < 1$ the graph $G'$ only has $\tOh(n)$ edges.  This bound crucially uses that the original graph is simple.  
The edge connectivity of $G$ is thus the minimum of $d$ and the edge connectivity of $G'$.
One can use Gabow's deterministic $\Oh(\lambda m' \log n)$ edge connectivity algorithm \cite{Gabow95} for a multigraph with $m'$ edges and edge connectivity $\lambda$  
to check in time $\tOh(n d \log n) = \tOh(m)$ if the edge connectivity of $G'$ is less than $d$ and, if so, compute it.    
In the most technical part of their work, Kawarabayashi and Thorup give a deterministic algorithm to find the $(1+\eps)$-KT partition of a simple graph $G$ in time $\tOh(m)$, 
giving an $\tOh(m)$ time deterministic algorithm overall for edge connectivity.
The key tool in their algorithm is the PageRank algorithm, which they use for finding low conductance cuts in the graph.

The KT partition has proven to be a very useful concept.  Rubinstein, Schramm, and Weinberg \cite{RSW18} also go through the $(1+\eps)$-KT partition to give a near-optimal $\tOh(n)$ randomized query algorithm determining the edge connectivity of 
a simple graph in the \emph{cut query} model.  In the cut query model one can query a subset of the vertices $S$ and receive in return the number of edges with exactly one endpoint in $S$.
En route to their result, \cite{RSW18} also improved the bound on the number of inter-component edges in the $(1+\eps)$-KT partition of a simple graph  to $\Oh(n)$, for any $\eps < 1$.  In the case $\eps =0$ this was independently 
done by Lo, Schmidt, and Thorup \cite{LST20}.  The KT partition approach is also used in the current best \emph{randomized} algorithm for edge connectivity, which runs in time 
$\Oh(\min\{m + n\log^2 n, m\log n\})$ \cite{GNT20}.\footnote{The bound quoted in \cite{GNT20} is $\Oh(m + n\log^3 n)$ but the improvement to Karger's algorithm by \cite{GawrychowskiMW20} reduces this to $\Oh(m + n\log^2 n)$.}

\subsection{Our results}
In this work we give the first near-linear time randomized algorithm to find the $(1+\eps)$-KT partition of a \emph{weighted} graph, for $0 \le \eps \le 1/16$.  An interesting aspect of our algorithm is that it uses Karger's 2-respecting cut framework to find the $(1+\eps)$-KT partition, thereby combining the aforementioned major lines of work on the minimum cut problem.
This makes progress on a number of problems.

\begin{enumerate}
\item
The polygon representation is a compact representation of the set of near-minimum cuts of a weighted graph, originally described by Bencz\'ur \cite{Benczur95,Ben97} and Bencz\'ur-Goemans \cite{BG08}.
It extends the cactus representation \cite{DKL76}, which only works for the set of exact minimum cuts, and has played a key role in recent breakthroughs on the traveling salesperson problem \cite{GSS11,KKG21}.
For a general weighted graph the polygon representation has size $O(n^2)$, and Bencz\'ur has given a randomized algorithm to construct a polygon representation of the $(1+\eps)$-near mincuts of a graph in time $\Oh(n^{2(1+\eps)})$ 
\cite[Section 6.3]{Ben97} by building on the Karger-Stein algorithm.
It is an open question whether we can construct a polygon representation in time $\tOh(n^2)$ for $\eps > 0$.  In his thesis \cite[pg.\ 126]{Ben97}, Bencz{\'u}r says, ``It already seems hard to directly identify the system of atoms within the $\tOh(n^2)$ time 
bound," where the system of atoms is defined analogously to the $(1+\eps)$-KT partition but for the set of all $(1+\eps)$-near minimum cuts, not just the non-trivial ones.  One can easily construct the set of atoms from a 
$(1+\eps)$-KT partition, thus our KT partition algorithm gives a $\tOh(m)$ time algorithm for this task as well, making progress on this open question.

\item
The $(1+\eps)$-KT partition of a weighted graph is exactly what is needed to give an optimal 
\emph{quantum} algorithm for minimum cut: Apers and Lee \cite{ApersL20} showed that the quantum query and time complexity of minimum cut in the adjacency matrix model 
is $\widetilde \Theta(n^{3/2} \sqrt{\tau})$ for a weighted graph where the ratio of the largest to smallest edge weights is $\tau$, with the algorithm proceeding by finding a $(1+\eps)$-KT partition.

In the case where the graph is instead represented as an adjacency list, they gave an algorithm with query complexity $\tOh(\sqrt{mn\tau})$ but whose running time is larger at 
$\tOh(\sqrt{mn\tau} + n^{3/2})$.  The bottleneck in the time complexity is the time taken to find a $(1+\eps)$-KT partition of a weighted graph with $\tOh(n)$ edges.  Using the near-linear time randomized algorithm we give 
to find a $(1+\eps)$-KT partition here improves the time complexity of this algorithm to $\tOh(\sqrt{mn\tau})$, matching the query complexity.
We detail the full algorithm in \cref{sec:quant-algo}.

Both quantum algorithms also used the following observation \cite[Lemma 2]{ApersL20}: if in a weighted graph $G$ the ratio of the largest edge weight to the smallest is $\tau$, then the total weight of inter-component edges in a $(1+\eps)$-KT partition of $G$ for $\eps < 1$ is $\Oh(\tau n)$, which can be tight.

\item
The best randomized algorithm to compute the edge connectivity of a simple graph is the 2-out contraction approach of Ghaffari, Nowicki, and Thorup \cite{GNT20}, 
which has running time $\Oh(\min\{m + n\log^2 n, m\log n\})$.  Using our algorithm to find a $(1+\eps)$-KT partition in a weighted graph we can follow
Karger's 2-respecting tree approach to compute the edge connectivity of a simple graph in time 
$\Oh(m + n \log^{6} n)$, thus also achieving the optimal bound on graphs that are not too sparse.  
We postpone details to \cref{sec:conn-algo}.
\end{enumerate}
Apart from these examples, we are hopeful that our near-optimal randomized algorithm for finding the KT partition of a weighted graph will find further applications.

In order to find a $(1+\eps)$-KT partition in near-linear time, Apers and Lee \cite{ApersL20} show that it suffices to solve the following problem in near-linear time.  
Let $G$ be a connected weighted graph and $T$ a spanning tree of $G$.  
Consider a graph $H$ whose nodes are the edges of $T$, and where two nodes $e,f$ of $H$ are connected by an edge iff the 2-respecting cut 
defined by $e,f$ is a non-trivial $(1+\eps)$-near minimum cut of $G$.  Then the problem is to find a spanning forest of $H$.  Our main technical 
contribution is to give a $\Oh(m \log^4 n)$ time deterministic algorithm to solve this problem, where $m$ is the number of edges of the original graph $G$.

It is interesting to compare the problem of finding a spanning forest of $H$ with the original problem solved by Karger of finding a minimum-weight 2-respecting cut of $T$.  
To find a spanning forest of $H$ we potentially have to find $\Omega(n)$ many $(1+\eps)$-near minimum cuts, which we accomplish with only an 
additional logarithmic overhead in the running time.  The first insight to how this might be possible is to note that Karger's original algorithm to find the minimum weight 2-respecting cut actually does something stronger 
than needed.  Let $\cost(e,f)$ be the weight of the 2-respecting cut of $T$ defined by $\{e,f\}$.  For \emph{every} edge $e$ of $T$ Karger's algorithm attempts to find an 
$f^* \in \argmin_f \cost(e,f)$.  It does not always succeed in this task, but if the candidate $f'$ returned for edge $e$ is not such a minimizer, then for $f^* \in \argmin_f \cost(e,f)$ it 
must be the case that the candidate $g$ returned for $f^*$ satisfies $\cost(f^*,g) \le \cost(e,f^*)$.  In this way, the algorithm still succeeds to find a minimum weight 2-respecting cut in the end.  

In contrast, we give an algorithm that for every edge $e$ of $T$ actually finds 
\[
f^* \in \argmin_f \large\{\cost(e,f): \{e,f\} \text{ defines a non-trivial cut} \large\} \enspace .
\] 
We then show that this suffices to implement a round of Bor\r{u}vka's spanning forest algorithm \cite{NMN01} on $H$ in near-linear time.  Bor\r{u}vka's spanning forest algorithm consists of $\log n$ rounds and maintains the 
invariant of having a partition $\{S_1, \ldots, S_k\}$ of the vertex set and a spanning tree for each set $S_i$.  The algorithm terminates when there is no outgoing 
edge from any set of the partition, at which point the collection of spanning trees for the sets of the partition is a spanning forest of $H$.  The sets of the partition are initialized to be individual 
nodes of $H$.  

In each round of Bor\r{u}vka's algorithm the goal is to find an outgoing edge from each set $S_i$ of the partition which has one.  Consider a node $e$ of $H$ with $e \in S_i$.  We can find the best 
partner $f$ for $e$ and check if $\{e,f\}$ indeed gives rise to a non-trivial $(1+\eps)$-near minimum cut and so is an edge of $H$.  The problem is that
$f$ could also be in $S_i$ in which case the edge $\{e,f\}$ is not an outgoing edge of $S_i$ as desired.
To handle this, we maintain a data structure that allows us to find both the best partner $f$ for $e$, but also the best partner $f'$ for $e$ that lies in a different set of the partition from $f$.  
We call this operation a \emph{categorical top two} query.  If 
there actually is an edge of $H$ with one endpoint $e$ and the other endpoint outside of $S_i$ then either $\{e,f\}$ or $\{e,f'\}$ will be such an edge.  Following the approach of \cite{GawrychowskiMW20} to the 
minimum-weight 2-respecting cut problem, combined with an efficient data structure for handling categorical top two queries, we are able to do this for all nodes $e$ of $H$ in near-linear time, which allows us to implement a 
round of Bor\r{u}vka's algorithm in near-linear time.

\subsection{Technical overview}\label{sec:technical}
We now give a more detailed description of our main result.  Let $G = (V,E,w)$ be a weighted graph, where $E$ is the set of edges and $w:E \rightarrow \R_+$ assigns 
a positive weight to each edge.  For a set $S \subseteq V$ let $\Delta_G(S)$ be the set of all edges of $G$ with exactly one endpoint in $S$.  A \emph{cut} of $G$ is a set of edges 
of the form $\Delta_G(S)$ for some $\emptyset \ne S \subsetneq V$.  We call $S$ and $\bar{S}$ the shores of the cut.  Let $w(\Delta_G(S)) = \sum_{e \in \Delta(S)} w(e)$.  
We use $\lambda(G) = \min_{\emptyset \ne S \subsetneq V} w(\Delta(S))$ for the minimum weight of a cut in~$G$.  

We will be interested in partitions of $V$ and the partial order on partitions induced by \emph{refinement}.  For two partitions $\cX,\cY$ of $V$ we say that $\cX \preceq \cY$ iff for every $X \in \cX$ there is a $Y \in \cY$ with $X \subseteq Y$.  
In this case we say $\cX$ is a \emph{refinement} of $\cY$.  The \emph{meet} of two partitions $\cX$ and $\cY$, denoted $\cX \wedge \cY$, is the partition $\cZ$ such that $\cZ \preceq \cX, \cZ \preceq \cY$ and for any other partition $\cW$ satisfying these 
two conditions $\cW \preceq \cZ$.  In other words, $\cX \wedge \cY$ is the greatest lower bound on $\cX$ and $\cY$ under $\preceq$.  For a set of partitions $\cD = \{\cD_1, \ldots, \cD_K\}$ we write 
$\bigwedge \cD = \cD_1 \wedge \cdots \wedge \cD_K$.

For our applications we need to consider not only minimum cuts, but also near-minimum cuts.  For $\eps \ge 0$, let $\cB_\eps(G) = \{ \{S, \bar{S}\} : w(\Delta_G(S)) \le (1+\eps)\lambda(G) \}$ be the set of 
all bipartitions of $V$ corresponding to $(1+\eps)$-near minimum cuts.  Let $\cB_\eps^{nt}(G) \subseteq \cB_\eps(G)$ be the set of all the non-trivial cuts in $\cB_\eps(G)$.  The $(1+\eps)$-KT partition of $G$ is 
exactly $\bigwedge \cB_\eps^{nt}(G)$.

Both $\bigwedge \cB_\eps(G)$ and $\bigwedge \cB_\eps^{nt}(G)$ are important sets for understanding the structure of (near)-minimum cuts in a graph.  Consider first $\bigwedge \cB_0(G)$, the meet of the set of all bipartitions corresponding to 
minimum cuts.  This set arises in the cactus decomposition of $G$ \cite{DKL76}, a compact representation of all minimum cuts of $G$.  
A cactus is a connected multigraph where every edge appears in exactly one cycle.  The edge connectivity of 
a cactus is 2 and the minimum cuts are obtained by removing any two edges from the same cycle.  A \emph{cactus decomposition} of a graph $G$ is a cactus $H$ on $\Oh(n)$ vertices and a mapping $\phi: V(G) \rightarrow V(H)$ 
such that $\Delta_G(\phi^{-1}(S))$ is a mincut of $G$ iff $\Delta_H(S)$ is a mincut of $H$.  The mapping $\phi$ does not have to be injective, so multiple vertices of $G$ can map to the same vertex of $H$.  In this case, 
however, the cactus decomposition property means that all vertices in $\phi^{-1}(\{v\})$ must be on the same side of every minimum cut of $G$, for every $v \in V(H)$.  Thus as $v$ ranges over $V(H)$ the 
sets $\phi^{-1}(\{v\})$ give the elements of $\bigwedge \cB_0(G)$ (note that $\phi^{-1}(\{v\})$ can also be empty).  A cactus decomposition of a weighted graph can be constructed by a randomized algorithm in near-linear time \cite{KP09}, 
thus this also gives a near-linear time randomized algorithm to compute $\bigwedge \cB_0(G)$.

Lo, Schmidt, and Thorup \cite{LST20} give a version of the cactus decomposition that only represents the non-trivial minimum cuts.  In fact, they give a deterministic $\Oh(n)$ time algorithm that converts a standard cactus 
into one representing the non-trivial minimum cuts.  Combining this with the near-linear time algorithm to compute a cactus decomposition, this gives a near-linear time randomized algorithm to compute $\bigwedge \cB_0^{nt}(G)$
as well.

The situation changes once we go to near-minimum cuts, which can no longer be represented by a cactus, but require the deformable polygon representation from \cite{Benczur95,Ben97,BG08}.  This construction is fairly intricate, and the best known randomized algorithm to construct a deformable polygon representation of the $(1+\eps)$-near mincuts of a graph builds on the Karger-Stein algorithm and takes time $\Oh(n^{2(1+\eps)})$ \cite[Section 6.3]{Ben97}.  A prerequisite to constructing a deformable polygon representation is being able to compute $\bigwedge \cB_\eps(G)$ as, analogously to the case of a cactus, these sets will be the ``atoms'' that label regions of the polygons.

Our main result in this work is to give a randomized algorithm to compute $\bigwedge \cB_\eps(G)$ and $\bigwedge \cB_\eps^{nt}(G)$ in time $\Oh(m \log^5 n)$.  

\begin{restatable}{theorem}{main}
\label{thm:main}
Let $G = (V,E,w)$ be a graph with $n$ vertices and $m$ edges.  For $0 \le \eps \le 1/16$ let $\cB_\eps = \{ \{S, \bar{S}\} : w(\Delta(S)) \le (1+\eps)\lambda(G) \}$ and $\cB_\eps^{nt} \subseteq \cB_\eps$ 
be the subset of $\cB_\eps$ containing only non-trivial cuts.  
Both $\bigwedge \cB_\eps$ and $\bigwedge \cB_\eps^{nt}$ can be computed with high probability by a randomized algorithm with running time $\Oh(m \log^5 n)$.
\end{restatable}

In the rest of this introduction we focus on computing $\bigwedge \cB_\eps^{nt}$.  It is easy to construct $\bigwedge \cB_\eps$ from $\bigwedge \cB_\eps^{nt}$ deterministically in $\Oh(n)$ time.

The first obstacle we face in designing a near-linear time algorithm to compute the meet of $\cB_\eps^{nt}$ is that the number of near-minimum cuts in $G$ can be $\Omega(n^2)$, so we cannot afford to 
consider all of them.  An idea to get around this is to try the following:
\begin{enumerate}
\item Efficiently find a ``small'' subset $\cB' \subseteq \cB_\eps^{nt}$ such that $\bigwedge \cB' = \bigwedge \cB_\eps^{nt}$.  We call such a subset a \emph{generating set}.
\end{enumerate}
A greedy argument shows that such a subset $\cB'$ exists of size at most $n-1$.  We initialize $\cB' = \{ \{S, \bar{S}\} \}$ for some element $\{S, \bar{S}\}$ in $\cB_\eps^{nt}$.  We then iterate through the elements 
$\{T, \bar{T}\}$ of $\cB_\eps^{nt}$ and add it to $\cB'$ iff $\bigwedge \cB' \cup \{T, \bar{T}\} \ne \bigwedge \cB'$.  Each bipartition added to $\cB'$ increases the number of elements in $\bigwedge \cB'$ by at least $1$.
As this size can be at most $n$, and begins with size $2$ the total number of sets at termination is at most $n-1$. 
This shows that a small generating set exists, but there still remains the problem of finding such a generating set \emph{efficiently}.

Assuming we get past the first obstacle, there remains a second obstacle.
The most straightforward algorithm to compute the meet of $k$ partitions of a set of size $n$ takes time $\Theta(kn \log n)$, which is 
again too slow if $k = \Theta(n)$.  Thus we will also need to 
\begin{enumerate}
\setcounter{enumi}{1}
\item Exploit the structure of $\cB'$ to compute $\bigwedge \cB'$ efficiently.
\end{enumerate}

Apers and Lee \cite{ApersL20} give an approach to accomplish (1) and (2) following Karger's framework of tree respecting cuts.
Karger shows that in near-linear time one can compute a set of $K \in \Oh(\log n)$ spanning trees $T_1, \ldots, T_K$ of $G$ such that every $(1+\eps)$-near minimum cut of $G$ 2-respects at least one of these trees.
Let $\cB_i \subseteq \cB_\eps^{nt}$ be the bipartitions corresponding to non-trivial near-minimum cuts that 2-respect $T_i$.  To compute $\bigwedge \cB_\eps^{nt}$ it suffices to compute 
$\cC_i = \bigwedge \cB_i$ for each $i=1, \ldots, K$ and then compute $\bigwedge_{i=1}^K \cC_i$.  The latter can be done in time $\Oh(n \log^2 n)$ by the aforementioned algorithm.  
This leaves the problem of computing $\bigwedge \cB_i$.

A key observation from \cite{ApersL20} gives a generating set $\cB_i'$ for $\cB_i$ of size $\Oh(n)$.  One intializes $\cB_i'$ to be empty and then adds bipartitions in $\cB_i$ that 1-respect $T_i$ to $\cB_i'$.  
This is a set of size $\Oh(n)$, and Karger has shown that all near-minimum cuts that 1-respect a tree can be found in time $\Oh(m)$.  

Now we focus on the cuts that strictly 2-respect $T_i$.  To handle these one creates a graph $H$ whose nodes are the \emph{edges} of $T_i$ and where there is an edge between 
nodes $e$ and $f$ iff the 2-respecting cut of $T_i$ defined by $\{e,f\}$ is a near-minimum cut in $\cB_i$.  One then adds to $\cB_i'$ the bipartitions corresponding to a set of 2-respecting 
cuts that form a spanning forest of $H$.  The resulting set $\cB_i'$ has size $\Oh(n)$ and it can be shown to be a generating set for $\cB_i$.  

Apers and Lee give a \emph{quantum} algorithm to find a spanning forest of $H$ with running time $\tOh(n^{3/2})$.  They then give a randomized algorithm to compute $\bigwedge \cB_i'$ 
in time $\tOh(n)$. 
As our main technical contribution, we give a deterministic algorithm to find a spanning forest of $H$ in time $\Oh(m \log^4 n)$.  
We also replace the randomization used in the algorithm to compute $\bigwedge \cB_i'$ with an appropriate data structure to give an $\tOh(n)$ deterministic algorithm to compute the meet.

\section{Preliminaries}
For a natural number $n$ we use $[n] = \{1,\ldots, n\}$.
\paragraph*{Graph notation}
For a set $S$ we let $S^{(2)}$ denote the set of unordered pairs of elements of $S$.
We represent an undirected edge-weighted graph as a triple $G = (V,E,w)$ where $E \subseteq V^{(2)}$ and $w: E \rightarrow \R_{+}$ gives the weight of an edge $e \in E$.  
We will also use $V(G)$ to denote the vertex set of $G$ and $E(G)$ to denote the set of edges.
We always use $n$ for the number of vertices in $G$ and $m$ for the number of edges.  
We will overload the function $w$ to let $w(F) = \sum_{e \in F} w(e)$ for a set of edges $F$ and for two disjoint sets $S,T \subseteq V$ we use $w(S,T)$
 to denote $\sum_{e \in E: |e \cap S| = |e \cap T| = 1} w(e)$, that is the sum of the weights of 
edges with one endpoint in $S$ and one endpoint in $T$.  For a subset $\emptyset \ne S \subsetneq V$ we let $\Delta(S)$ be the set of edges with 
exactly one endpoint in $S$.  This is the \emph{cut} defined by $S$.
We let $\lambda(G)$ denote the weight of a minimum cut in $G$, i.e., $\lambda(G) = \min_{\emptyset \neq S \subsetneq V} w(\Delta(S))$.  

\paragraph*{Heavy path decomposition}
We use the standard notion of \emph{heavy path decomposition} of $T$ \cite{ST1983,HT1984}, which is a partition of
the edges of $T$ into \emph{heavy paths}.
We define this partition recursively: first,
find the heavy path starting at the root by repeatedly descending to the child of the current node with the largest subtree.
This creates the topmost heavy path starting at the root (called its head) and terminating at a leaf (called its tail). Second,
remove the topmost heavy path from $T$ and repeat the reasoning on each of the obtained smaller trees. The crucial property
is that, for any node $u$, the path from $u$ to the root in $T$ intersects at most $\log n$ heavy paths.

\paragraph*{Algorithmic preliminaries}
We collect here a few theorems from previous work that we will need.
The first is Karger's result \cite{Karger} about finding $\Oh(\log n)$ many spanning trees of a graph $G$ such that every minimum cut of 
$G$ will 2-respect at least one of these trees.  We will need the easy extension of this result to near-minimum cuts, which has been explicitly stated by Apers and Lee.
\begin{theorem}[{\cite[Theorem 4.1]{Karger}, \cite[Theorem 24]{ApersL20}}] 
\label{thm:karger1}
Let $G$ be a weighted graph with $n$ vertices and $m$ edges.  There is a randomized algorithm that in time 
$\Oh(m \log^2n + n\log^4n)$ time constructs a set of $\Oh(\log n)$ spanning trees such that every $(1+1/16)$-near minimum cut of $G$ 2-respects $1/4$ of them with high probability.
\end{theorem}

We will also need the fact that for a weighted graph $G=(V,E,w)$ the values in $G$ of all 1-respecting cuts of a tree $T$ can be computed quickly.  For a rooted spanning tree $T$ of $G$ and an edge $e \in E(T)$, let 
$T_e$ be the set of vertices in the component not containing the root when $e$ is removed from $T$.  We use the shorthand $\cost(e) = w(\Delta_G(T_e))$.
\begin{lemma}[{\cite[Lemma 5.1]{Karger}}]
\label{lem:one_respect}
Let $G$ be a weighted graph with $n$ vertices and $m$ edges, and $T$ a spanning tree of $G$.
There is a deterministic algorithm that computes $\cost(e)$ for every $e \in E(T)$ in time $\Oh(m+n)$.
\end{lemma}

We will also make use of the improvement by Gawrychowski, Mozes and Weimann of Karger's mincut algorithm.
\begin{lemma}[{\cite[Theorem 7]{GawrychowskiMW20}}]
\label{lem:mlogn}
Let $G$ be a weighted graph with $n$ vertices and $m$ edges and $T$ a spanning tree of $G$.
A cut of minimum weight in $G$ that 2-respects $T$ can be found deterministically in $\Oh(m \log n)$ time.  Using this, there is a randomized algorithm that finds a minimum cut in $G$ with high probability in time $\Oh(m \log^2 n)$.
\end{lemma}

Finally we give the formal statement of the result from \cite{ApersL20} that underlies our algorithm to construct a KT partition.
\begin{lemma}[{\cite[Lemma 29]{ApersL20}}]
\label{lem:ALstree}
Let $T = (V,E)$ be a tree and $\cM \subseteq 2^V$ a family of subsets of $V$ such that $|\Delta_T(S)| = 2$ for each $S \in \cM$.  
Let $Q = \{ \Delta_T(S) : S \in \cM\}$ be a set of pairs of edges in $E$.  Suppose $F$ is spanning forest 
for the graph $L = (E,Q)$.  Then the set of shores of the 2-respecting cuts defined by edges in $E(F)$ is a generating set for $\bigwedge \cM$.
\end{lemma}

\section{Data structures} \label{sec:data-structures}
In this section we develop the data structure we will need for a fast implementation of our spanning forest algorithm.  We 
want to maintain a tree $T$ with root $r$, in which each edge has a score and a color, so that we can support the following queries and updates.  For an edge $e$ of the tree, let $T_e$ be the set of edges 
in the component not containing $r$ when $e$ is removed from the tree.  On query an edge $e$ we want to find the edge $f \in T_e$ with the 
smallest score, and the edge $f' \in T_e$ with the smallest score among edges whose color is different from that of $f$.  We call such a query 
a categorical top two query.  We want to answer these queries while allowing adding $\Delta$ to the score of every edge on the path between two nodes.
We could use link-cut trees \cite{ST1983} to accomplish this with $\Oh(\log n)$ update and query time using the fact that 
link-cut trees can be modified to support any semigroup operation under path updates.
However, in our case the tree is static, and this allows for a simple and self-contained solution
that requires only a well-known binary tree data structure coupled with the standard
heavy path decomposition of a tree.
This comes at the expense of implementing updates in $\Oh(\log^{2}n)$ time instead of $\Oh(\log n)$ time.
The construction can be seen as folklore and has been explicitly stated by Bhardwaj, Lovett and Sandlund~\cite{BLS20} for the case when each edge
maintains its score and there are no colors.
We provide a detailed description of such an approach in \cref{app:data-structure}.
We note that the increased update time does not to affect the overall time complexity of our algorithm.

\begin{restatable}{lemma}{toptwopath}
\label{lem:top2path}
Let $A[1],\ldots, A[n]$ be an array where each element has two fields, a color $A[i].\edgecolor$ and a score $A[i].\score$.  In $\Oh(n)$ time we can create a 
data structure using $\Oh(n)$ space and supporting the following operations in $\Oh(\log n)$ time per operation.
\begin{enumerate}
\item $\textsc{Add}(\Delta,i,j)$: for all $i \le k \le j$ do $A[k].\score \gets A[k].\score + \Delta$,
\item $\textsc{CatTopTwo}(i,j)$: return $(k_1, k_2)$ where $k_1 = \argmin \{A[k].\score : i \le k \le j\}$ and $k_2 = \NULL$ if $A[k].\edgecolor = A[k_1].\edgecolor$ for 
all $i \le k \le j$ and $k_2 = \argmin \{A[k].\score : i \le k \le j, A[k].\edgecolor \ne A[k_1].\edgecolor\}$ otherwise.  
\end{enumerate}
\end{restatable}

\begin{restatable}{lemma}{toptwotree}
\label{lem:top2tree}
Let $T$ be a tree on $n$ nodes, with each edge $e\in T$ having its color and score. In $\Oh(n)$ time we can create a data structure
using $\Oh(n)$ space and supporting the following operations.
\begin{enumerate}
\item $\textsc{AddPath}(\Delta,p)$: add $\Delta$ to the score of every edge on a path $p$ in $T$ in $\Oh(\log^{2} n)$ time.
\item $\textsc{CatTopTwo}(e)$: categorical top-two query in $T_{e}$ in $\Oh(\log n)$ time.
\end{enumerate}
\end{restatable}

\section{Spanning tree of near-minimum 2-respecting cuts in near-linear time} \label{sec:sp-tree-algorithm}
Let $G = (V,E,w)$ be a weighted undirected graph.  We will assume throughout that $G$ is connected, and 
in particular that $m \ge n-1$, as the KT partition of a disconnected graph can be easily determined from its connected components.
Let $T$ be a spanning tree of $G$.  We will choose an $r \in V$ with degree 1 in $T$ to be the root of $T$. We view $T$ as a directed graph with all 
edges directed away from $r$.  With some abuse of notation, we will also use $T$ to refer to this directed version.  
If we remove any edge $e \in E(T)$ from $T$ then $T$ becomes disconnected into two components.  We use $e^\downarrow \subseteq V$ to denote the 
set of vertices in the component not containing the root, and $T_e \subseteq E(T)$ to denote the set of edges in the subtree rooted at the head of $e$, 
i.e.\ the edges in the subgraph of $T$ induced by $e^\downarrow$.  We further use the shorthand $\cost(e) = w(\Delta(e^\downarrow))$ for the weight of the cut with shore $e^\downarrow$.  

Two edges $e,f \in E(T)$ define a unique cut in $G$ which we denote by $\cut_T(e,f)$ (or $\cut(e,f)$ if it is clear from the context which $T$ we are referring to).  The cut depends on the relationship between $e$ and $f$.  If $e \in T_{f}$ or $f \in T_e$
then we say that $e$ and $f$ are \emph{descendant} edges.  Without loss of generality, say that $f \in T_e$.  Then the cut 
defined by $e$ and $f$ is $\cut(e,f) = \Delta(e^\downarrow \setminus f^\downarrow)$. 
If $e$ and $f$ are not descendant edges, then we say they are \emph{independent}.  For independent edges we see that $\cut(e,f) = \Delta(e^\downarrow \cup f^\downarrow)$.  In both cases we use $\cost(e,f)$ to denote the weight of the corresponding cut.

In a KT partition we are only interested in non-trivial cuts.  We first prove the following simple claim that characterizes when $\cut(e,f)$ is trivial.  
\begin{proposition}
\label{clm:trivial}
Let $G = (V,E,w)$ be a connected graph with $n$ vertices and let $T$ be a spanning tree of $G$ with root $r$.  For $e,f \in E(T)$ if $\cut(e,f)$ is trivial then 
\begin{enumerate}
	\item If $e,f$ are independent then they must be the unique edges incident to the root.
	\item If $e,f$ are descendant then there is a vertex $v \in V$ such that $e$ is the edge incoming to $v$ and $f$ is the unique edge outgoing from $v$, or vice versa.
\end{enumerate}
\end{proposition}

\begin{proof}
First suppose that $e,f$ are independent.  Then a shore of $\cut(e,f)$ is $S = e^\downarrow \cup f^\downarrow$.  We have that $|g^\downarrow| \ge 1$ any $g \in E(T)$, thus $|S| \ge 2$.  
Hence for $\cut(e,f)$ to be trivial we must have $|\bar{S}| = 1$.  The root $r$ is not contained in $g^\downarrow$ for any $g \in E(T)$ thus it must be the case that $\bar{S} = \{r\}$.  
For this to happen, $e,f$ must be incident to $r$, and $r$ cannot have any other outgoing edges besides $e$ and $f$.  

Now consider the case that $e,f$ are descendant and suppose without loss of generality that $f \in T_e$.  Let $S = e^\downarrow \setminus f^\downarrow$.  In this case we have 
$|S| < n-1$ as $e^\downarrow$ does not contain the root and $|f^\downarrow| \ge 1$.  Let us understand when $|S| = 1$.  As all vertices on the path from the head of $e$ to and including the 
tail of $f$ are in $S$ it must be the case that the head of $e$ is the tail of $f$.  Call this vertex $v$ and note $v \in S$.  If $v$ has any other child $u$ besides the head of $f$ then we would have $u \in S$ as well, 
thus $f$ must be the unique outgoing edge from $v$.
\end{proof}
By choosing a root $r$ for $T$ that has degree 1 we avoid the case of item~1 of \cref{clm:trivial}.  Thus we only have to worry about trivial cuts when $e,f$ are descendant.

With that out of the way, we now turn to the main theorem of this section.
As outlined in \cref{sec:technical} this theorem is the key routine in our $(1+\eps)$-KT partition algorithm, which we fully describe in \cref{sec:algorithm}.
\begin{restatable}{theorem}{sptree}
\label{thm:sp-tree}
Let $G = (V,E,w)$ be a connected graph with $n$ vertices and $m$ edges and let $T$ be a spanning tree of $G$.  For a given parameter $\beta$, define the graph 
$H$, with $V(H) = E(T)$ and $E(H)=\{ \{e,f\} \in E(T)^{(2)} : \cost(e,f) \leq \beta \text{ and } \cut(e,f) \text{ non-trivial}\}$.  There is a deterministic algorithm that given adjacency list access to $G$ and $T$ 
outputs a spanning forest of $H$ in $\Oh(m\log^{4}n)$ time.  
\end{restatable}

At a high level, we prove \cref{thm:sp-tree} by following Bor\r{u}vka's algorithm to find a spanning forest of $H$. 
Throughout the algorithm we maintain a subgraph $F$ of $H$ that is a forest, initialized to be the empty graph on vertex set $E(T)$.  At the end of the algorithm, $F$ will be a spanning forest of $H$.  The algorithm
proceeds in rounds. In each round, for every tree in the forest, we find an edge connecting it to another tree in the forest, if
such an edge exists.  If $H$ has $k$ connected components, then in each round the number
of trees in $F$ minus $k$ goes down by at least a factor of $2$, and so the algorithm terminates in $\Oh(\log n)$ rounds.

The main work is implementing a round of Bor\r{u}vka's algorithm.  We will think of the nodes of $F$ as having colors, where nodes in the same tree of the forest have the 
same color, and nodes in distinct trees have distinct colors.  
The goal of a single round is to find, for each color $c$, a pair of edges $e,f \in T$ such that $c=\edgecolor(e) \ne \edgecolor(f)$ and $\{e,f\} \in E(H)$,
or detect that there is no such pair with these properties, in which case the nodes colored $c$ in $F$ already form a connected component of $H$.  As we need to refer to such 
pairs often we make the following definition.
\begin{definition}[partner]
Let $T$ and $H$ be as in \cref{thm:sp-tree}.
Given an assignment of colors to the edges of $T$ we say that $f$ is a partner for $e$ if $\{e,f\} \in E(H)$ and $\edgecolor(e) \ne \edgecolor(f)$.
\end{definition}

We will actually do something stronger than what is required to implement a round of Bor\r{u}vka's algorithm, which we encapsulate in the following code header.
\begin{algorithm}[H]
\caption{$\Round$}
\label{alg:mindeg}
 \hspace*{\algorithmicindent} \textbf{Input:} Adjacency list access to $G$, a spanning tree $T$ of $G$, a parameter $\beta$, and an assignment of colors to each $e \in E(T)$. \\
 \hspace*{\algorithmicindent} \textbf{Output:} For every $e \in E(T)$ output a partner $f \in E(T)$, or report that no partner for $e$ exists.
\end{algorithm}

The implementation of $\Round$ is our main technical contribution.  Let us first see how to use $\Round$ to find a spanning forest of $H$. 
\begin{lemma}
\label{lem:boruvka}
Let $G$, $T$ and $H$ be as in \cref{thm:sp-tree}.
There is a deterministic algorithm that makes $\Oh(\log n)$ calls to $\Round$ 
and in $\Oh(n \log n)$ additional time outputs a spanning forest of $H$.
\end{lemma}

\begin{proof}
We construct a spanning forest of $H$ by maintaining a collection of trees $F$ that will be updated in rounds by Bor\r{u}vka's algorithm until 
it becomes a spanning forest.  We initialize $F = (E(T), \emptyset)$ and give all $e \in E(T)$ distinct colors.  
We maintain the invariants that $F$ is a forest and that nodes in the same tree have the same color and those in different trees 
have distinct colors.

Consider a generic round where $F$ contains $q$ trees.  We call $\Round$ with the current color assignment.  For every $e$ which has one we obtain a partner $f$ such that 
$\{e,f\} \in E(H)$ and $\edgecolor(e) \ne \edgecolor(f)$. For each color class $c$ we select one $e$ with $\edgecolor(e) = c$ which has a returned partner (if it exists) and let $X$ be the set of selected edges.  
We then find a maximal subset of edges $X' \subseteq X$ that do not create a cycle among the color classes by computing a spanning forest of the graph whose supervertices are given by the color classes 
and edges given by $X$. 
We add the edges in $X'$ to $E(F)$.
Finally we merge the color classes of the connected components in $F$ by appropriately updating the color assignments, and we pass the updated forest and color assignments to the next round of the algorithm.
Each of the steps in a single round can be executed in $\Oh(n)$ time.

We have that $|X'| \ge (q - \cc(H))/2$ where $\cc(H)$ is the number of connected components of $H$.  Each edge from $X'$ added to $F$ decreases the number of trees in $F$
by one.  Thus the number of trees in $F$ minus $\cc(H)$ decreases by at least a factor of 2 in each round and the algorithm terminates after $\Oh(\log n)$ rounds.  The time spent 
outside of the calls to $\Round$ is $\Oh(n)$ for each of the $\Oh(\log n)$ rounds.
This is $\Oh(n \log n)$ overall.
\end{proof}

If a node $e$ has a partner $f$, then $\{e,f\}$ can either be a pair of descendant or independent edges.  
To implement $\Round$ we will separately handle these cases, as described in the next two subsections.

\subsection{Descendant edges}
\label{sub:descendant}

We follow the approach from~\cite{GawrychowskiMW20} originally designed to find a single pair $\{e,f\}$ of descendant edges that minimizes $\cost(e,f)$ over all $e,f \in E(T)$ in $\Oh(m \log n)$ time.  
Their approach actually does something stronger (as does Karger's original algorithm): for every $e \in E(T)$ it finds the best match in the subtree $T_e$, i.e., it returns an edge $f^* \in \argmin \{ \cost(e,f) \mid f \in T_e\}$.  
In order to implement the descendant edge part of $\Round$ we have three additional complications to handle:
\begin{enumerate}
  \item The edge $f^*$ might have the same color as $e$.
  \item The resulting $\cut(e,f^*)$ might be a trivial cut.
  \item Edge $e$ may have no partner in $T_e$ but still have a partner $f$ such that $e \in T_f$.  This partnership may not be discovered when we are looking for partners of $f$ if there is another $g \in T_f$ with $\cost(f,g) \le \cost(e,f)$.
\end{enumerate}

Item 1 can be easily solved by, in addition to finding $f^*$, also finding $g^* \in \argmin\{ \cost(e,f) \mid f \in E(T_e),\, \edgecolor(f) \ne \edgecolor(f^*) \}$.
Phrasing things in this way, rather than simply looking for the edge $h$ with color different from $e$ which minimizes $\cost(e,h)$, helps to limit the dependence of the query on $e$ and thus reduce the query time.
If there is an $f \in T_e$ with $\edgecolor(f) \ne \edgecolor(e)$ and $\cost(e,f) \le \beta$ then at least one of $f^*, g^*$ will 
satisfy this too.

For item 2, we use the result of \cref{clm:trivial} that descendant edges that give rise to trivial cuts have a very constrained structure.  This allow us to avoid trivial cuts when looking 
for a partner of~$e$.

Item~3 is relatively subtle and does not arise in the minimum weight 2-respecting cut problem.  To explain the issue we have to first say something about the high level structure of 
our implementation of $\Round$.  We will perform an Euler tour of $T$ and, when the tour visits edge $e$ for the first time, we will look for a partner $f$ for $e$ in $T_e$.  
The issue is the following, which we explain in the context of the very first round of Bor\r{u}vka's algorithm 
so we do not have to worry about nodes having different colors.  Suppose that in the graph $H$ the only edge incident 
to node $e$ is a node $f$ with $e \in T_f$.  Thus in the execution of $\Round$ we want to find $f$ as a partner of $e$.  When the Euler tour is at $e$ we will not find any suitable 
partner for $e$, as there is none in $T_e$.  We would like to identify $f$ as a partner for $e$ when the Euler tour visits $f$ for the first time. 
However, if there is a $g \in T_f$ with $\cost(f,g) < \cost(f,e)$ then the algorithm will return $g$ as a partner of $f$ rather than $e$.  To handle this we will actually make two 
passes over $T$.  In the first pass, when we visit edge $e$ for the first time we look for a partner $f$ in $T_e$.  In the second pass, we handle the case where the partner of $e$ might be an ancestor of $e$.
To do this we need to de-activate nodes.  When the Euler tour visits $f$ for the first time, we first find the lowest cost partner for $f$ in $T_f$.
We then de-activate this node, and again find the best active partner for $f$ in $T_f$.  Repeating this process, we will eventually find $e$ if $\{e,f\}$ is indeed an edge of $H$ and $e,f$ have 
different colors.

Now we turn to more specific implementation details.  A key idea in~\cite{GawrychowskiMW20} is that we can do an Euler tour of $T$ while maintaining a data structure such that when we first visit an edge $e$ we can easily look up $\cost(e,f)$ for any $f \in T_e$.  The way this is maintained can be best understood by noting that for $f \in T_e$:
\begin{align} \label{eq:score}
\cost(e,f) &= w(\Delta(e^\downarrow \setminus f^\downarrow)) \nonumber\\
&= w(e^\downarrow \setminus f^\downarrow, (e^\downarrow)^c) + w(e^\downarrow \setminus f^\downarrow, f^\downarrow) \nonumber\\
&= \cost(e) + \underbrace{\cost(f) - 2 w(f^\downarrow, (e^\downarrow)^c)}_{\score_e(f)} \enspace ,
\end{align}
where for convenience we defined $\score_e(f) = \cost(f) - 2 w(f^\downarrow, (e^\downarrow)^c)$, where the superscript $c$ denotes taking the complement.

We begin the algorithm by computing $\cost(e)$ for every $e \in E(T)$, which can be done in $\Oh(m)$ time by \cref{lem:one_respect}.  
We then do an Euler of $T$ while maintaining a data structure from \cref{lem:top2tree} such that, when we are considering $e \in E(T)$, for every
$f \in T_e$ the value of the data structure at location $f$ is $\cost(e,f)$.  For $f \not \in T_e$ this will not in general be the case.

As can be seen from \cref{eq:score}, the key to maintaining this data structure is how to update the values $w(f^\downarrow, (e^\downarrow)^c)$ when we descend edge $e$.
Consider the case where we are currently at edge $e' = (z,x)$ and move to a descending edge $e = (x,y)$.  For two vertices $u,v$ let $p(u,v)$ be the set of edges on the path from $u$ to $v$ in $T$, and let $\lca(u,v)$ be their least common ancestor in $T$.
For $f \in T_e$ we see that
\begin{equation}
w(f^\downarrow, (e^\downarrow)^c) = w(f^\downarrow, (e'^\downarrow)^c) + \sum_{\{u,v\} \in E \atop f \in p(u,v), \lca(u,v) = x} w(\{u,v\}) \enspace.
\end{equation}
By its definition in \eqref{eq:score} we can compute $\score_e(f)$ from $\score_{e'}(f)$ by \emph{subtracting} $2w(\{u,v\})$ from for every $\{u,v\} \in E$ such that $f \in p(u,v)$ and $\lca(u,v) = x$.
The for-loop on line 2 of \cref{alg:dfs_shell} implements this step for all $f$ by looping over all $\{u,v\} \in E$ with $\lca(u,v) = x$.
After this update we have that $\cost(e,f) = \cost(e) + \score(f)$ for every $f \in T_e$.
This shows how to descend down $T$ while keeping the invariant.
The full tree is then explored by taking an Euler tour through $T$, and whenever we go back up in the tree we revert the score updates (for-loop on line 10 of \cref{alg:descendant}).
This allows us to find candidate $f\in T_{e}$ for every $e\in E(T)$.
To bound the number of updates, note that each of the $m$ edges has a unique lca, and we only do an update corresponding to an edge when the lca is visited by the Euler tour.
Since the Euler tour visits every vertex at most twice, the number of updates is at most $2m$.
In addition, the number of categorical top two queries is $n-1$.
The data structure from Lemma~\ref{lem:top2tree} then yields $\Oh(m\log^{2} n)$ time overall.

The algorithm is formalized in \cref{alg:dfs_shell}, whose correctness we prove in the following lemma.

\begin{lemma}[{cf.\ \cite[Lemma 8]{GawrychowskiMW20}}]
\label{lem:ET}
Assume that we first initialize $e.\score \leftarrow \cost(e)$ for every $e \in E(T)$, and then run \cref{alg:dfs_shell} (doing nothing in line \ref{line:process}).
Then whenever an edge $e = (x,y)$ is followed on line~\ref{line:forloop} in the call to \textsc{Traverse}$(x)$ 
it holds that $\cost(e,f) = \cost(e) + f.\score$ for all~$f \in T_e$.
\end{lemma}

\begin{proof}
We will prove this by induction on the depth of $x$.  Consider the case where $x$ is the root $r$.  Before the call to \textsc{Traverse}$(r)$ we initialized all scores to $e.\score  \gets \cost(e)$.  Then, on line~\ref{line:update} of \textsc{Traverse}$(r)$, for each $\{u,v\} \in E(G)$ with $\lca(u,v) = r$ we subtract $2w(\{u,v\})$ from the score 
of every edge on the $u$ to $v$ path in $T$.  Let us refer to scores at this point in time as ``at time zero.'' We first claim that at time zero for any outgoing edge $e = (r,y)$ from the root this makes $\cost(e,f) = \cost(e) + f.\score$ for all $f \in T_e$.  

Let $p(u,v)$ be the set of edges on the path from $u$ to $v$ in $T$.
By \cref{eq:score} we have $\cost(e,f) = \cost(e) + \cost(f) - 2w(f^\downarrow, (e^\downarrow)^c)$ thus it suffices to show that for any $f \in T_e$
\[
w(f^\downarrow, (e^\downarrow)^c) = \sum_{\{u,v\} \in E(G) \atop f \in p(u,v), \lca(u,v) = r} w(\{u,v\}) \enspace .
\]
This holds because by definition $h\in E(f^\downarrow, (e^\downarrow)^c)$ iff one endpoint is in $f^\downarrow$ the other endpoint is in $(e^\downarrow)^c)$, which in turn happens 
iff the least common ancestor of the endpoints is $r$ and $f$ lies on the path between the endpoints.  

To finish the base case, we claim that at each iteration of the for loop all scores are the same as at time zero.  This is because in the recursive calls that follow the update to the scores on line~\ref{line:update} 
the update is exactly canceled out by the reverse update on line~\ref{line:revert} when the recursive call exits.

For the inductive step, let us suppose that when an edge $e = (x,y)$ is followed on line~\ref{line:forloop} in the call to \textsc{Traverse}$(x)$ 
it holds that $\cost(e,f) = \cost(e) + f.\score$ for all $f \in T_e$.  Let us now refer to scores at this point in time as ``at time zero.''  We then want to show that on line~\ref{line:forloop} in the call to \textsc{Traverse}$(y)$ that for an outgoing 
edge $e' = (y,z)$ it holds that $\cost(e',f) = \cost(e') + f.\score$ for all $f \in T_{e'}$.  The change in the scores from time zero to the execution of the for loop in the call to \textsc{Traverse}$(y)$ occurs in 
the update on line~\ref{line:update}.  Let us refer to scores at this point in time as ``at time one.''  We first show that at time one for any outgoing edge $e' = (y,z)$ of $y$ it holds that $\cost(e',f) = \cost(e') + f.\score$ for all $f \in T_{e'}$.
The key to this is to consider the difference between $\cost(e,f)$ and $\cost(e',f)$ for an $f \in T_{e'}$.  By \cref{eq:score} we have $\cost(e',f) = \cost(e') + \cost(f) - 2 w(f^\downarrow, (e'^\downarrow)^c)$, and 
by the inductive hypothesis at time zero $\score.f = \cost(f) - 2 w(f^\downarrow, (e^\downarrow)^c)$.  Thus so that $\cost(e',f) = \cost(e') + \score.f$ we need to update the $\score.f$ by 
\begin{equation}
\label{eq:ind_update}
2 \left(w(f^\downarrow, (e^\downarrow)^c) - w(f^\downarrow, (e'^\downarrow)^c) \right) = -2 \sum_{\{u,v\} \in E(G) \atop f \in p(u,v), \lca(u,v) = y} w(\{u,v\}) \enspace.
\end{equation}
To see this, first note that $E(f^\downarrow, (e^\downarrow)^c) \subseteq E(f^\downarrow, (e'^\downarrow)^c)$.  This confirms that we should subtract something to perform this update.  
An edge $\{u,v\}$ is in $E(f^\downarrow, (e'^\downarrow)^c)$ but not $E(f^\downarrow, (e^\downarrow)^c)$ iff one endpoint, say $u$, is in $f\downarrow$ and the other endpoint $v$ is 
in $(e'^\downarrow)^c \setminus (e^\downarrow)^c$.  This means that $v \in e^\downarrow$ but now $e'^\downarrow$ and so $y = \lca(u,v)$.  The condition $u \in f^\downarrow$ is then 
equivalent to having $f$ on the path between $u$ and $v$.  This confirms that \cref{eq:ind_update} performs the correct update.

To finish the proof, we claim that $\cost(e',f) = \cost(e') + \score.f$ not just at time one, but at the time when the for loop with $e'$ is executed.  This is again because in the changes to the scores 
on line~\ref{line:update} that are made in a recursive call are reversed when the recursive call exits on line~\ref{line:revert}, thus every time the for loop is executed the scores are the same 
as the scores at time one.
\end{proof}

\begin{algorithm}[!htbp]
\caption{Euler tour maintaining $\cost(e,f)$ for $f \in T_e$}
\label{alg:dfs_shell}
\begin{algorithmic}[1]
\Function{Traverse}{$x$}
\ForAll{$\{u,v\} \in E(G)$ such that $\lca(u,v)=x$}
  \State \textsc{AddPath}($-2 w(\{u,v\}),u\text{-to-}v$) \label{line:update}
\EndFor
\ForAll{$y$ such that $e = (x,y) \in E(T)$} \label{line:forloop}
  \State Process $e$. \Comment{``Process'' depends on context.} \label{line:process}
  \State \Call{Traverse}{$y$}
\EndFor
\ForAll{$\{u,v\} \in E(G)$ such that $\lca(u,v)=x$}
   \State \textsc{AddPath}($2 w(\{u,v\}),u\text{-to-}v$)  \label{line:revert}
\EndFor
\EndFunction
\end{algorithmic}
\end{algorithm}

Given \cref{lem:ET} to maintain $\cost(e,f)$ for $f \in T_e$ during an Euler tour of the tree, and with a link cut tree data structure to handle 
categorical top two queries, it is now straightforward to design an algorithm to find for every edge $e$ a partner for $e$ that is a descendant or 
ancestor, if such a partner exists.  

\begin{algorithm}[!htbp]
\caption{The descendant edge portion of $\Round$.}
\label{alg:descendant}
\begin{algorithmic}[1]
\State $X \gets \emptyset$ \Comment{$X$ will hold all edges found during the round}
\For{$e \in E(T)$} $e.\score  \gets \cost(e)$ \Comment{scores are maintained with \cref{lem:top2tree}} \label{line:score_init}
\EndFor  
\State Run \Call{Traverse}{$r$} where ``Process $e$'' means running \Call{Below}{$e$}.
\State Run \Call{Traverse}{$r$} where ``Process $e$'' means running \Call{Above}{$e$}.
\end{algorithmic}
\vspace{2mm}
\begin{algorithmic}[1]
\Function{Below}{$e$} \Comment{Find a partner for $e$ in $T_e$.}
\If{the head of $e$ has outdegree 1}
  \State Let $h$ be the outgoing edge of the head of $e$. \Comment{In this case $\cut(e,h)$ is trivial}
  \State $h.\score \pluseq \beta + 1$
\EndIf
 \State $(f,g) =$ \Call{CatTopTwo}{$e$}
 \If{$\edgecolor(f) \ne \edgecolor(e) \And \cost(e) + \score(f) \le \beta$}  
         \State $X \gets X \cup \{e,f\}$
 \ElsIf{$\edgecolor(g) \ne \edgecolor(e) \And \cost(e) + \score(f) \le \beta$}
        \State $X \gets X \cup \{e,g\}$
 \EndIf
\If{the head of $e$ has outdegree 1}
  \State $h.\score \minuseq \beta + 1$
\EndIf
\EndFunction
\end{algorithmic}
\vspace{2mm}
\begin{algorithmic}[1]
\Function{Above}{$e$} \Comment{Find all $f$ such that $f \in T_e$ and $e$ is a partner of $f$.}
\If{the head of $e$ has outdegree 1}
  \State Let $h$ be the edge coming into the head of $e$. \Comment{In this case $\cut(e,h)$ is trivial}
  \State $h.\score \pluseq \beta + 1$
\EndIf
\State noMore = \False
\Repeat \label{line:repeat}
 \State $(f,g) =$ \Call{CatTopTwo}{$e$}
 \If{$\edgecolor(f) \ne \edgecolor(e) \And \cost(e) + \score(f) \le \beta$}  
         \State $X \gets X \cup \{e,f\}$
         \State $f.\score \pluseq \beta + 1$
 \ElsIf{$\edgecolor(g) \ne \edgecolor(e) \And \cost(e) + \score(f) \le \beta$}
        \State $X \gets X \cup \{e,g\}$
        \State $g.\score \pluseq \beta + 1$
 \Else 
 	\State noMore = True
 \EndIf
\Until noMore
\If{the head of $e$ has outdegree 1}
  \State $h.\score \minuseq \beta + 1$
\EndIf
\EndFunction
\end{algorithmic}
\end{algorithm}

\begin{theorem}
\label{thm:descendant}
Given an assignment $e.\edgecolor$ for each $e \in E(T)$, there is a deterministic algorithm that runs in time $\Oh(m \log^{2} n)$ and for each $e$ finds 
an $f$ such that 
\begin{enumerate}
  \item $\{e,f\} \in H$
  \item $e \in T_f$ or $f \in T_e$
  \item $e.\edgecolor \ne f.\edgecolor$
\end{enumerate}
if such an $f$ exists.
\end{theorem}

\begin{proof}
The algorithm is given by \cref{alg:descendant}.  Suppose that an edge $e$ has a partner $f$ satisfying the 3 conditions of the theorem.  Then 
either $f \in T_e$ or $e \in T_f$.  We claim that if $f \in T_e$ then we will find a partner of $e$ in the call to \textsc{Traverse}($r$)
using \textsc{Below}($e$) to process edge $e$,
and if $e \in T_f$ 
then we will find a partner of $e$ in the call to \textsc{Traverse}($r$) using \textsc{Above}($e$) to process edge $e$.

Let us show these statements separately, starting with the case $f \in T_e$.  Consider the time when $e$ is considered in the for loop line~\ref{line:forloop} 
in a recursive call from \textsc{Traverse}($r$) using \textsc{Below}($e$) to process edge $e$.   In the call to 
\textsc{Below}($e$) we first check if the tail of $e$ has a single outgoing edge $h$.  If this is the case then $\cut(e,h)$ is a trivial cut and thus we do not want to 
find $h$ as a partner for $e$.  We thus add $\beta+1$ to the score of $h$ ensuring that it will never be a valid partner for $e$.  By \cref{clm:trivial} this is the only situation
where $\cut(e,f)$ is trivial for $f \in T_e$.
By \cref{lem:ET} for all other $g \in T_e$ it holds that $\cost(e,g) = \cost(e) + g.\score$. 
Thus the call to CatTopTwo($e$) will perform correctly, and one of the returned edges must be a valid partner for $e$.  We then reset the score of $h$, if it was changed, to 
maintain the property given by \cref{lem:ET}.

Now consider the case where $e$ has a partner $f$ with $e \in T_f$.  Let $f$ be the first such partner that is encountered in an Euler tour of $T$.  We claim 
that the edge $\{e,f\}$ will be added to $X$ in the call to \textsc{Traverse}($r$) using \textsc{Above}($e$) to process edge $e$.  First note that after the previous call to \textsc{Traverse}($r$) terminates it holds that 
$g.\score = \cost(g)$ as all changes to the scores in the recursive calls are reverted after the call returns.  Thus we are again in position to apply \cref{lem:ET}, 
although we have to be slightly more careful this time as scores are modified within the body of the for loop on line \ref{line:forloop} of \cref{alg:dfs_shell} when we run \textsc{Above}($e$) to process $e$.  We again handle 
the possibility of trivial cuts as in the ``below'' case.   We also modify a score for an edge $e$ after a partner for $e$ has been found and thus our job for $e$ is done and we no longer need to use its score.  As by 
assumption $f$ is the first potential partner for $f$ encountered in the Euler tour, the score of $e$ has not been modified at this point.  Thus by \cref{lem:ET} 
it holds that $\cost(e,f) = \cost(f) + e.\score$.  This means that $e$ will eventually be found in the repeat loop on line~\ref{line:repeat} of the function \textsc{Above}($e$).

Let us now analyze the running time.  Computing $\cost(e)$ for each edge $e$ can be done in time $\Oh(m + n)$ by \cref{lem:one_respect}.
Before proceeding with the traversal, we gather, for each node $x$, all edges $\{u,v\} \in E(G)$ such that $\lca(u,v)=x$.
This can be done in $\Oh(m)$ time by constructing in $\Oh(n)$ time a constant-time LCA structure~\cite{BenderF00}, and iterating over the edges.
Next consider the body of the for loop in the \textsc{Below} function.  Here we make a single \textsc{CatTopTwo} query which takes 
time $\Oh(\log n)$ by \cref{lem:top2tree}.  Thus over the entire Euler tour these queries contribute $\Oh(n \log n)$ to the running time.  In the \textsc{Above} function
all \textsc{CatTopTwo} queries in the body of the for loop except for one (when noMore becomes true) will result in de-activating an edge.  Thus again the total 
query time over the Euler tour is $\Oh(n \log n)$.

Finally consider the cost of updating the scores in the Euler tour.  As discussed earlier, over the course of the Euler tour this requires doing 2 calls to \textsc{AddPath} for every edge of $G$.  Each \textsc{AddPath} call can be done in time $\Oh(\log^{2} n)$ by \cref{lem:top2tree}, thus the overall time for this is $\Oh(m \log^{2} n)$, which dominates the complexity of the algorithm.
\end{proof}

\subsection{Independent edges}
\label{sub:independent}
The goal now is to find, for every edge $e\in E(T)$, a partner $f \in E(T)$ such that $e,f$ are independent, or decide that there is no such $f$.  As we chose the root of $T$ to have degree~1, by \cref{clm:trivial} we do 
not have to worry about trivial cuts in the independent edge case.  Instead of considering all edges $e\in E(T)$ one-by-one, we first find a heavy path decomposition of $T$ and then iterate over 
all pairs of heavy paths $h,h'$ to look for a partner $f \in h'$ for every $e \in h$.  We cannot literally carry out this plan as the number of pairs of heavy paths 
can be $\Omega(n^2)$ and so we cannot explicitly consider every pair. We show next that many pairs $h,h'$ result in a trivial case 
and that all these trivial pairs can be solved together in one batch.  We then bound the number of non-trivial pairs and show that in near-linear time we can explicitly process all of them.
The idea of processing pairs of heavy paths, and explicitly considering only the non-trivial ones, was introduced in the context of 2-respecting cuts by
Mukhopadhyay and Nanongkai~\cite{MN20} (see also~\cite{GawrychowskiMW20}).

Consider two distinct heavy paths $h,h'$, where $h$ is the path $u_1-u_2-\cdots-u_q$ and $h'$ is the path $v_{1}-v_{2}-\ldots v_{q'}$.  We let $e_i = (u_i, u_{i+1})$ for 
$i = 1, \ldots, q-1$ and $f_i = (v_i, v_{i+1})$ for $i = 1, \ldots, q'-1$.  It can be that not all pairs $e_i, f_j$ are independent, see \cref{fig:rectangle}.  However, we can easily 
identify the subpaths of $h,h'$ containing pairwise independent edges in constant time by computing the least common ancestor $v$ of the tails of $h,h'$.  If $v=v_{p'}$ 
lies on $h'$ then $e_i, f_j$ will be independent for $1\le i < q$ and $p' \le j <q'$, and similarly if $v$ lies on $h$.  In general we assume that $p,p'$ have been determined 
so that $e_i, f_j$ are independent for all $p' \le i < q$ and $p' \le j <q'$, and that these pairs comprise all of the independent pairs on $h,h'$.  
We can associate to $h,h'$ a $(q-1)$-by-$(q'-1)$ matrix $\M$ where for $p' \le i < q$ and $p' \le j <q'$
\begin{align}
\label{eq:score_ind}
\M[i,j]
&= \cost(e_i, f_j) \nonumber \\
&= \cost(e_i) + \cost(f_j)- 2 w(e_i^\downarrow, f_j^\downarrow) \enspace ,
\end{align}
and $\M$ is undefined otherwise.\footnote{We could restrict $\M$ to the submatrix on which it is defined, but find it notationally easier for the 
$i,j$ indices in $\M$ to match the edge labels.} 
By \cref{lem:one_respect}, all values of $\cost(e)$ can be computed in $\Oh(m)$ total time.
To efficiently evaluate $\M$, we will prepare a list $L(h,h')$ of all edges that contribute to $w(e^\downarrow, f^\downarrow)$ for independent $e,f$ with $e\in h, f\in h'$.  
For many $h,h'$ the list $L(h,h')$ will be empty, leading to the trivial case mentioned above.
The following lemma bounds the size of all the non-empty lists and shows they can be constructed efficiently.

\begin{lemma}
\label{lem:lists}
The total length of all lists $L(h,h')$ is $\Oh(m \log^2 n)$ and all non-empty lists $L(h,h')$ can be constructed deterministically in time $\Oh(m \log^2 n)$.
\end{lemma}

\begin{proof}
Observe that an edge $\{u,v\}\in E$ can contribute to $w(e^\downarrow,f^\downarrow)$ for independent $e,f$ with $e\in h, f\in h'$ only if
$u$ is in the subtree rooted at the head of $h$ and $v$ is in the subtree rooted at the head of $h'$.
There are at most $\log n$ heavy paths intersecting
the path from $u$ to the root and from $v$ to the root, and we can iterate over all such heavy paths in time proportional to
their number (for example, by storing for each edge of $T$, a pointer to the head of the heavy path that contains it). Thus, we can
iterate over all relevant $h,h'$ in $\Oh(\log^{2}n)$ time, and add a triple $(h,h',\{u,v\})$ to an auxiliary list in which the heavy paths are identified by their heads. 
The total size of the auxiliary list is now $\Oh(m\log^2 n)$ and it can be lexicographically sorted in the same time with radix sort. After sorting each non-empty list $L(h,h')$ constitutes a contiguous fragment of the auxiliary list.
\end{proof}

We can now describe how to find a partner $f$ for every $e$ such that $e,f$ are independent.
The algorithm first solves together in one batch the case where the partner of $e \in h$ is in a heavy path $h'$ where $L(h,h')$ is empty.
After that we explicitly consider all $h,h'$ with $L(h,h')$ non-empty.  We consider these two cases in the next two subsections.

\subsubsection{Empty lists}

\begin{lemma}
\label{lem:empty}
There is a deterministic algorithm that in time $\Oh(m + n)$ finds a partner for every edge $e \in E(T)$ that has a partner $f$ such that $e,f$ are independent and $e \in h, f \in h'$ with $L(h,h')$ empty.
\end{lemma}

\begin{proof}
The key observation is that if $L(h,h')$ is empty then $\cost(e,f) = \cost(e) + \cost(f)$ by \cref{eq:score_ind}.  As can be seen from \cref{eq:score} and \cref{eq:score_ind}, for any edge $f'$ 
it always holds that $\cost(e,f') \le \cost(e) + \cost(f')$, whether $e,f'$ are descendant or independent.  
Thus in this case it suffices for us to find \emph{any} $f'$ of color 
different from $e$ such that $\cost(e) + \cost(f') \le \beta$, and $\cut(e,f')$ is non-trivial as this ensures $\cost(e,f') \le \beta$.  We are guaranteed such an $f'$ exists as $f$ satisfies this.

By \cref{lem:one_respect} we can compute $\cost(f')$ for every $f' \in E(T)$ in time $\Oh(m)$.  Then in time $\Oh(n)$ with one pass over $E(T)$ we compute the edge $f_1$ of lowest cost and 
the edge $f_2$ of lowest cost that is of color different to $f_1$.  We then repeat this categorical top two query twice more, each time excluding all previously found edges.  At the end we obtain edges 
$f_1, \ldots, f_6$.  We claim that for every $e$, at least one of these must be a valid partner.  

Consider any particular $e$.
The first categorical top two query can only fail to find a valid partner for $e$ if one of $f_1, f_2$ creates a trivial cut with $e$.  In this case, the second categorical top two query 
can only fail if one of $f_3,f_4$ creates a trivial cut with $e$ as well.  By \cref{clm:trivial}, however, there are at most two possible edges that can create a trivial cut with $e$, thus in this case the 
third categorical top two query must succeed and we find a valid partner for $e$.
\end{proof}

\subsubsection{Non-empty lists}
The more difficult case is to find partners among pairs $h,h'$ with $L(h,h')$ non-empty.  To solve this case 
we will use the special structure of $\M$.  As above, say that $h$ is the path $u_1-u_2-\cdots-u_q$ and $h'$ is the 
path $v_{1}-v_{2}-\ldots v_{q'}$, and let $e_i = (u_i, u_{i+1})$ for $i = 1, \ldots, q-1$ and 
$f_i = (v_i, v_{i+1})$ for $i = 1, \ldots, q'-1$.  Further suppose $e_i,f_j$ are independent for all $p\le i < q, p' \le j < q'$.  
We have that $\M[i,j] = \cost(e_i) + \cost(f_j) -2w(e_i^\downarrow,f_j^\downarrow)$ 
for $p\le i < q, p' \le j < q'$.
Recall that $L(h,h')$ is defined precisely as the list of edges that contribute to $w(e^\downarrow,f'^\downarrow)$ for independent $e \in h, f' \in h'$.
The contribution of a specific edge $\{u,v\} \in L(h,h')$ can be understood as follows: let $u_{i}$ be the lowest common ancestor of $u$ and $u_{q}$, and $v_{j}$ be the lowest common ancestor of $v$ and $v_{q'}$.
Then the weight of $\{u,v\}$ contributes to $M[a,b]$ for every $p \le a \leq i$, $p' \leq b \leq j$.
This is depicted in \cref{fig:rectangle}.  We will compute these indices $i$ and $j$ for every $\{u,v\} \in L(h,h')$.
This takes constant time per edge using an appropriate LCA structure~\cite{BenderF00}, and so total time $\Oh(|L(h,h')|)$.
Let $\cL(h,h') = \{(i,j) \mid \{u,v\} \in L(h,h') , u_i = \lca(u,u_q), v_j = \lca(v,v_{q'})\}$ denote the resulting list of index pairs, each of which has an associated weight.

\begin{figure}[h]
\includegraphics[width=\textwidth]{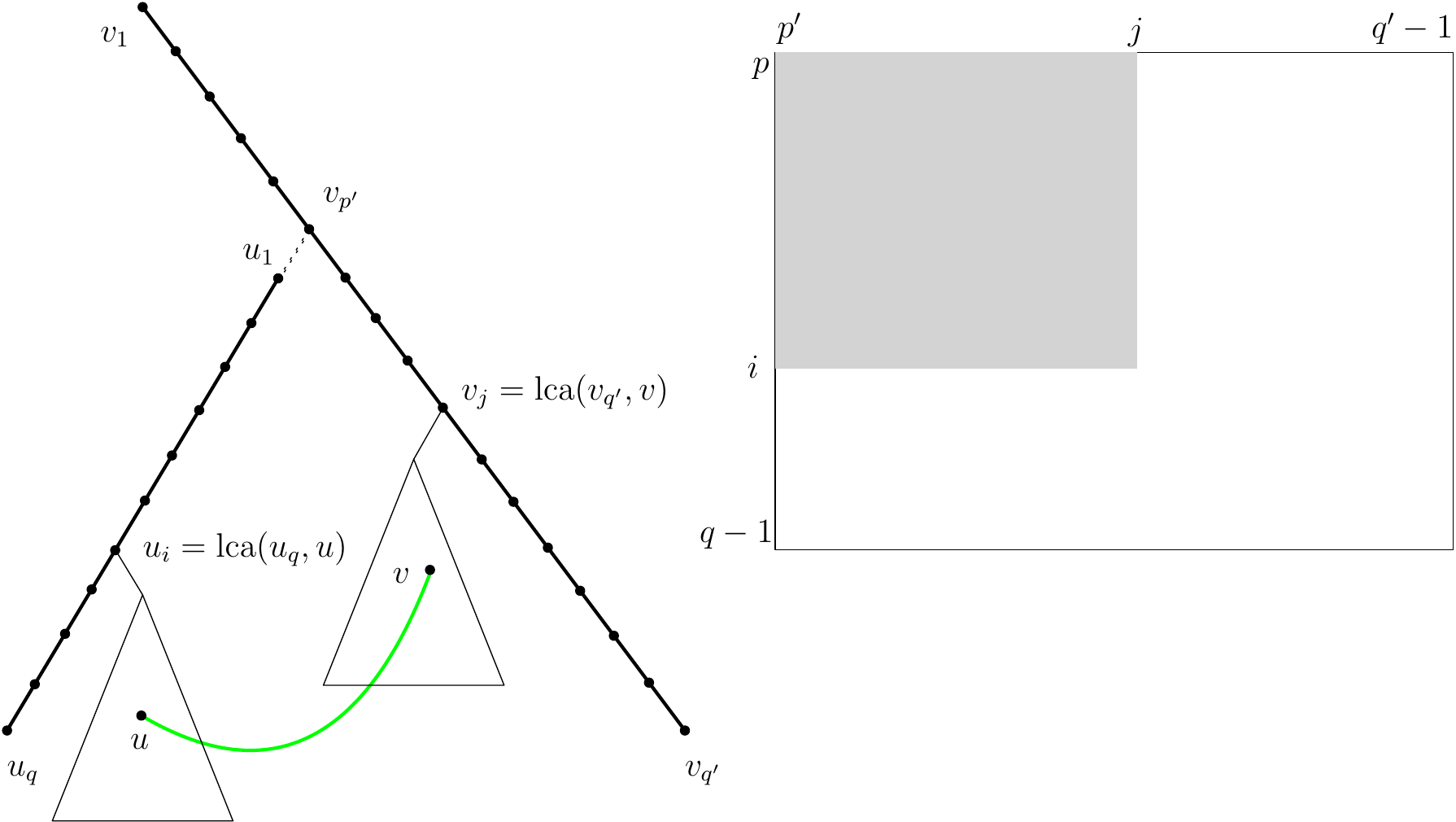}
\caption{Contribution of an edge $\{u,v\} \in L(h,h')$ (denoted in green on the left) to $M^{(h,h')}[\cdot,\cdot]$ (denoted in grey on the right).}
\label{fig:rectangle}
\end{figure}

\begin{lemma}
\label{lem:ind_non-empty}
Let $\cF = \{ e \mid \exists h,h',f: e \in h, f \in h', e,f \text{ are partners and } L(h,h') \text{ non-empty}\}$.  There is a deterministic algorithm 
to find a partner for every $e \in \cF$ in time $\Oh(m \log^3 n)$.
\end{lemma}

\begin{proof}
The algorithm is given in \cref{alg:independent-non-empty}.  We describe the algorithm here and analyze its correctness and running time.

For every heavy path $h$ let $A_h$ be an array with $A_h[e].\score = \cost(e)$ and $A_h[e].\edgecolor = \edgecolor(e)$ for every $e \in h$.  Via \cref{lem:top2path} there is a data structure 
that supports path updates and \textsc{CatTopTwo} queries to $A_h$ in $\Oh(\log n)$ time.  The total time for this initialization step is $\Oh(n)$.  

Let $L$ be an ordered list of pairs that contains $(h,h')$ and $(h',h)$ for every $h,h'$ with $L(h,h')$ non-empty.  We sort $L$ by the name of the first path with radix sort in $\Oh(m \log^2 n)$ time.  
We will follow $L$ to iterate over all $h,h'$ with $L(h,h')$ non-empty.

Let us describe what the algorithm does when considering pairs $h,h'$ where $h$ consists of edges $e_1, \ldots, e_q$ and $h'$ consists of the edges $f_1, \ldots, f_{q'}$, where 
edges $e_i,f_j$ are independent for $p \le i < q, p' \le j < q'$, and these comprise all the independent pairs in $h,h'$.  
We iterate over the columns of $\M$, starting from $q'-1$ and going until $p'$, and maintain the invariant that, when considering column $j$, it holds that $A_h[i].\score = \M[i,j] - \cost(f_j)$ for every \emph{active} edge with index $p \le i < q$ (where active will be defined later).  We postpone describing how to maintain this invariant for the moment.  Then we do a \textsc{CatTopTwo} query on $A_h$ which returns potential candidates $e_a, e_b$.  If there is an edge $e \in h$ for which $f_j$ is a valid partner then 
$f_j$ must be a partner for either $e_a, e_b$.  This can be checked in constant time.  If $f_j$ is not a partner for either then we move on to column $j-1$; if it is a partner for, say $e_a$, then 
we add $\beta+1$ to $A_h[a].\score$ to ``de-activate'' $e_a$ and repeat the process by doing a \textsc{CatTopTwo} query again on $A_h$ until no valid partner is returned.

The basic algorithm we have described considers every column of $\M$ from $q'-1$ to $p'$.  We now show how to accelerate this algorithm by restricting our attention to a subset of the columns of $\M$ in this interval.  
Let $K_{h,h'} = |L(h,h')|$.  
We sort the pairs in $\cL(h,h')$ by the second coordinate in time $\Oh(K_{h,h'}\log K_{h,h'})=\Oh(K_{h,h'}\log n)$.
Let $J_1 < \cdots < J_t$ be the distinct values of the second coordinate that appear in this sorted list, where $t \le K_{h,h'}$.  Set $J_0 = p'-1$ and $J_{t+1} = q'-1$.  For $J_k < j \le J_{k+1}$ we have 
that $w(e_i^\downarrow, f_j^\downarrow)$ is constant over $j$ for every $p \le i < q$ by the definition of $\cL(h,h')$.  We call such an interval a void interval.  Thus the minimum of $\M[i,j]$ over $J_k < j \le J_{k+1}$ necessarily 
occurs at $j^* = \argmin_{J_k < j \le J_{k+1}} \cost(j)$.  This means that if edge $e_i \in h$ has a partner $f_j$ with $J_k < j \le J_{k+1}$ then one of the two edges returned by a \textsc{CatTopTwo}$(J_k+1, J_{k+1})$ to $A_{h'}$ 
must be a partner for $e_i$.   

We can thus amend the algorithm to the following.  For $J = J_{t+1}, J_t, \ldots, J_1$ we iterate over the endpoints of the void intervals.  When $J= J_k$ we maintain the invariant that $A_h[i].\score = \M[i,J_k] - \cost(f_{J_k})$.  
Thus for all $J_{k-1} < j \le J_{k}$ it holds that $\M[i,j] = A_h[i].\score + \cost(f_j)$.  We then do a \textsc{CatTopTwo} query on $A_h[p:q-1]$ and a \textsc{CatTopTwo} query on $A_{h'}$ with the interval $(J_{k-1}, J_k]$.  If any $e \in h$ has a 
partner $f_j$ with $J_{k-1} < j \le J_{k}$ at least one of the 4 possible pairs returned must be partners.  We de-activate any $e \in h$ which finds a partner by adding $\beta+1$ to its score and repeat the process until no valid partners 
are found, at which point we move on to the next void interval.  If $P$ partners are found then the total time spent in this void interval will be $\Oh((P+1) \log n)$ for the \textsc{CatTopTwo} queries and updates to de-activate edges.   

It remains to describe how to maintain the invariant $A_h[i].\score = \M[i,J_k] - \cost(f_{J_k})$ for all $p \le i < q-1$ when $J=J_k$.  To do this for every pair $(i,j)$ in $\cL(h,h')$ with $j = J_k$ with associated weight $w$ we 
subtract $2w$ from $A_h$ in the interval $[p,J_k]$.  Each such interval update can be done in time $\Oh(\log n)$ by \cref{lem:top2path} so the total time for all updates is $\Oh(K_{h,h'} \log n)$.  

Once we finish processing $h'$, we reverse all of the interval updates (but not the edge de-activations) so that we again have $A_h[i].\score = \cost(e_i)$ for all $i$ for all active edges $e_i$.  
This again can be done in time $\Oh(K_{h,h'} \log n)$.  Once we finish processing all pairs $h'$ associated $h$ we subtract $\beta + 1$ from $A_h[e].\score$ for all edges $e$ that were de-activated 
to make them active again.

The total number of edge de-activations is at most $n$ thus this contributes $\Oh(n \log n)$ to the running time and is low order.  
Over all $h,h'$ the total time spent is $\Oh(m \log^3 n)$.  
\end{proof}

\begin{algorithm}[!htbp]
\caption{Find partners among non-empty lists}
\label{alg:independent-non-empty}
\begin{algorithmic}[1]
\ForAll{heavy paths $h$}
  \State
  Initialize data structure $A_h$ with scores $\cost(e)$ and colors $e.\edgecolor$ for all $e \in h$.
\EndFor
\State Compile a list $L$ of all ordered pairs $(h,h')$ and $(h',h)$ with $L(h,h')$ non-empty.  Sort by the name of the first path.
\ForAll{$h$ that appears as a first path in $L$}
  \ForAll{$h'$ paired with $h$ in $L$}
    \State
    Compute the index list $\cL(h,h')$ with associated weights.
    \State
    Partition the interval $[p',q'-1]$ into void intervals $(J_k, J_{k+1}]$ for $k=0,\ldots, t$.
    \State Set $\Found=\emptyset$.
    \ForAll{$k=t$ to $0$} 
      \State 
      For all $(i',J_k) \in \cL(h,h')$, subtract twice its weight from $A_h[p:i'].\score$.
      \State
      Call \textsc{CatTopTwo}($J_k+1,J_{k+1}]$) on $A_{h'}$ to obtain edges $f_c, f_d$.
      \While{\True}
        \State Call \textsc{CatTopTwo}($p,q-1)$ on $A_{h}$ to obtain edges $e_a,e_b$.
        \If{$f_c$ or $f_d$ is a partner for $e_a$} 
        		\State Save this pair, add $e_a$ to $\Found$, and do $A_h[e_a].\score \pluseq \beta+1$.
        \ElsIf{$f_c$ or $f_d$ is a partner for $e_b$} 
        		\State Save this pair, add $e_b$ to $\Found$ and do $A_h[e_b].\score \pluseq \beta+1$.
        \Else \State break
        \EndIf
      \EndWhile
    \EndFor
    \ForAll{$k = t$ to $0$}
    \State
    For all $(i',J_k) \in \cL(h,h')$, add twice its weight to $A_h[p:i'].\score$.
    \EndFor
  \EndFor
  \State Subtract $\beta+1$ from the score of all edges in $\Found$.
\EndFor
\end{algorithmic}
\end{algorithm}

\subsection{Spanning tree algorithm}
We now have all components of the spanning tree algorithm, which we can combine together to implement $\Round$.

\begin{lemma}
\label{lem:round_imp}
There is a deterministic algorithm to implement $\Round$ which runs in time $\Oh(m \log^3 n)$.  
\end{lemma}

\begin{proof}
Given an assignment of colors to the edges of $T$, our task is to find a partner for every $e \in E(T)$ which has one.  If $e$ has a partner 
$f$ such that $e,f$ are in a descendant relationship then a partner for $e$ can be found in time $\Oh(m \log^{2} n)$ by \cref{thm:descendant}.  
The other case is that $e$ has a partner $f$ such that $e,f$ are independent.  This divides into two subcases.  If the heavy paths $h,h'$ 
containing $e,f$ respectively are such that $L(h,h')$ is empty then we will find a partner for $e$ via \cref{lem:empty} in time $\Oh(m)$.  
The bottleneck of the algorithm is where $L(h,h')$ in non-empty in which case we use \cref{lem:ind_non-empty} to find a partner in 
time $\Oh(m \log^3 n)$.  
\end{proof}

We can now prove the main theorem of this section, \cref{thm:sp-tree}, that we can find a spanning tree of $H$ in time $\Oh(m \log^4 n)$.
\begin{proof}[Proof of \cref{thm:sp-tree}]
Follows from \cref{lem:boruvka} and \cref{lem:round_imp}.
\end{proof}

\section{KT partition algorithm} \label{sec:algorithm}
For completeness we state here the full KT partition algorithm, including the reductions from \cite{ApersL20}.  At a high level, we follow Karger's algorithm to find 
$\Oh(\log n)$ spanning trees so that with high probability every $(1+\eps)$-minimum cut 2-respects at least one of them.  We then use our algorithm from \cref{thm:sp-tree}
to, for each tree $T$, find a generating set for the meet of all $(1+\eps)$-minimum cuts that 2-respect $T$.  We are then left with two problems.  
The first is that we still have to find the meet of the partitions in the generating set.  A near-linear time randomized algorithm was given to do this in 
\cite{ApersL20}.  Here we give a deterministic algorithm to do this.  Then we need to take the meet of $\Oh(\log n)$ partitions, one for each tree.  This is 
simple to do and we handle this first.

\begin{lemma}
\label{lem:par_meet}
Let $\{\cS_1, \ldots, \cS_K\}$ be a set of $K$ partitions of $[n]$.  There is a deterministic algorithm to compute $\bigwedge_{i=1}^K \cS_i$ in time $\Oh(Kn \log n)$.
\end{lemma}

\begin{proof}
In $\Oh(Kn \log n)$ time we can assign each $j \in [n]$ a $\Oh(K \log n)$ bit key indicating which set contains $j$ in each $\cS_i$.  
Collecting together elements with the same key gives $\bigwedge_{i=1}^K \cS_i$. 
\end{proof}

Next we see how to explicitly construct the meet of all $(1+\eps)$-minimum cuts that 2-respect a tree from a generating set.
We follow the idea of the proof in \cite{ApersL20} but make it deterministic by replacing random hashing with an appropriate data structure.

\begin{lemma}[Mehlhorn, Sundar, and Uhrig~\cite{MehlhornSU97}]
\label{lem:strings}
A dynamic family of persistent sequences, each of length at most $n$, can be maintained under the following
updates. A length-1 sequence is created in constant time, and a new sequence can be obtained  by joining and splitting
existing sequences in $\Oh(\log n(\log U\log^{*}U+\log n))$ time, where $U$ is the number of updates executed so far.
Each sequence $s$ has an associated signature $\mathrm{sig}(s)\in [U^{3}]$ with the property that $s=s'$ iff $\mathrm{sig}(s)=\mathrm{sig}(s')$.
\end{lemma}

For the proof of the lemma it will be useful to use the following definition.
\begin{definition}[separate]
Let $V$ be a finite set and $X \subseteq V$.  For $u,v \in V$ we say that $X$ separates $u,v$ if exactly one of 
them is in $X$.
\end{definition}

\begin{lemma}[cf.\ {\cite[Lemma 31]{ApersL20}}] \label{lem:euler-tour}
Consider as input a tree $T$ on a vertex set $V$ of size $n$, and sets of edge singletons $\cC_1 \subseteq E(T)$ and edge pairs $\cC_2 \subseteq E(T)^{(2)}$.
These sets define sets of 1-respecting and 2-respecting cuts, respectively.
There is an algorithm that in time $\Oh((n+|\cC_1| + |\cC_2|)\log^{2}n\log^{*}n)$ returns the meet of the bipartitions induced by these cuts.
\end{lemma}

\begin{proof}
We root the tree at an arbitrary vertex $r \in V$.  When we speak of the shore of a cut we always refer to the shore not containing $r$.  Arrange all elements of $\cC_1$
and $\cC_2$ in an arbitrary order to obtain a sequence of $N=|\cC_1| + |\cC_2|$ cuts.  Our goal is to construct, for each node $v \in V$, a string $s(v) \in \{0,1\}^N$ 
where the $\ith$ bit of $s(v)$ is $1$ iff the shore of the $\ith$ cut contains $v$. 
Assuming that we can indeed efficiently construct such strings, the meet is obtained by grouping together nodes $v$ with the
same string $s(v)$. However, the difficulty is that we cannot afford to construct $s(v)$ for all $v$ explicitly as this would require $nN$ bits. Instead, we will
use \cref{lem:strings} for representing a collection of strings of length $N$.  

Consider the preorder traversal of $T$ starting from the root $r$.  By definition $s(r) = 0^N$, which we create in the data structure by $N$ joins of $0$.  We then 
create $s(v)$ from the string $s(\parent(v))$ during the preorder traversal, where $\parent(v)$ 
is the parent of $v$.  To do this we set $s(v) \gets s(\parent(v))$ and then flip the bits of $s(v)$ corresponding to cuts whose shore contains $v$ but not $\parent(v)$ or vice versa.  
Thus we need to understand when the shore of a 2-respecting cut separates $v$ from $\parent(v)$.  The shore of a one respecting cut defined by edge $e$ is $e^{\downarrow}$,
and hence separates $v$ and $\parent(v)$ iff $e = \{v, \parent(v)\}$.  A 2-respecting cut defined by edges $\{e,f\}$ separates two vertices $u$ and $v$ iff exactly one of $e,f$ 
is on the path from $u$ to $v$ in $T$.  Thus a  2-respecting cut will separate $v$ and $\parent(v)$ iff either $e = \{v, \parent(v)\}$ or $f = \{v, \parent(v)\}$.  Hence there will be 
at most $2N$ bit flips in total and in $\Oh(N)$ time we can annotate the tree with which bits should be flipped at each node.  

A bit flip can be implemented in the data structure by a constant number of splits, joins and the creation of a length-1 sequence.  As there are 
$\Oh(n+N)$ total operations on the data structure, the total time for all updates is $\Oh((n+N) \log^2 n  \log^* n)$ by \cref{lem:strings}. 

Having obtained all the strings $s(v)$, we can group together nodes $v$ with the same string $s(v)$ by sorting their signatures $\mathrm{sig}(s(v))$.
Because each signature is a positive integer bounded by $\Oh(N^{3})$ by \cref{lem:strings}, this can be implemented with radix sort in $\Oh(N)$ time.
This gives the lemma.
\end{proof}

\begin{algorithm}[!htbp]
\caption{$(1+\eps)$-KT partition}
\label{alg:KT-partition}
 \hspace*{\algorithmicindent} \textbf{Input:} A weighted graph $G = (V,E,w)$ an a parameter $0 \le \eps \le 1/16$ \\
 \hspace*{\algorithmicindent} \textbf{Output:} $(1+\eps)$-KT partition
\begin{algorithmic}[1]
\State
Use Karger's tree packing algorithm to construct a set of $K \in \Oh(\log n)$ spanning trees $\cT = \{T_1, \ldots, T_k\}$ so that with high probability every $(1+\eps)$-minimum cut 2-respects at least one of them (\cref{thm:karger1}).
\State
Compute the weight of a minimum 2-respecting cut for each tree in $\cT$ by \cref{lem:mlogn}, and let $\lambda$ be the minimum value found.
\For{$i = 1,2,\dots,K$}
  \State Find the set $\cA_i = \{ e \in E(T_i) \mid \cost(e) \leq (1+\eps) \lambda \}$ indexing the 1-respecting near-minimum cuts of $T_i$ by \cref{lem:one_respect}.
  \State Use \cref{thm:sp-tree} with tree $T_i$ to find a spanning forest $T_{H_i}$ of the graph $H_i$ with edge set $\{ \{e,f\} : e,f \in E(T_i), \cost(e,f) \le \beta, \cut(e,f) \text{ non-trivial} \}$.  Set $\cB_i = \{ \{e,f\} \in E(T_{H_i}) \}$.
  \State Use \cref{lem:euler-tour} to construct the partition $\cS_i$ induced by the cuts indexed by $\cA_i$ and $\cB_i$. \label{line:meetTi}
\EndFor
\State
Output the meet $\cS = \bigwedge_{i=1}^K \cS_i$ by \cref{lem:par_meet}. \label{line:par_meet}
\end{algorithmic}
\end{algorithm}

We are now ready to prove our main theorem.
\main*
\begin{proof}
We first prove the theorem for $\bigwedge \cB_\eps^{nt}$.
The algorithm for computing $\bigwedge \cB_\eps^{nt}$ is given in \cref{alg:KT-partition}.
Let us first argue the correctness.  Step~1 succeeds with high probability by \cref{thm:karger1}, and the rest of the algorithm is deterministic.  Thus if we show that the algorithm is correct 
assuming that Step~1 succeeds, then the algorithm will be correct with high probability.  

Let us now assume that Step~1 succeeds.  Then $\lambda = \lambda(G)$ in Step~2 .  Let 
$\cT_i$ be the set of bipartitions of all non-trivial $(1+\eps)$-minimum cuts of $G$ that 2-respect $T_i$, for $i = 1, \ldots, K$.  We have that $\cup_i \cT_i = \cB_\eps^{nt}$.  Therefore
\[
\bigwedge \cB_\eps^{nt} = \bigwedge_{i=1}^K \bigwedge \cT_i \enspace.
\]
For each $i$ we have $\bigwedge (\cA_i \cup \cB_i) = \bigwedge \cT_i$ by the correctness of our main algorithm \cref{thm:sp-tree} and \cref{lem:ALstree}.  We compute 
$\bigwedge (\cA_i \cup \cB_i)$ via \cref{lem:euler-tour}.  Finally, we compute 
$\bigwedge_{i=1}^K \bigwedge \cT_i$ in Step~\ref{line:par_meet} by \cref{lem:par_meet}.

Now let us go over the time complexity.  Step~1 runs in time $\Oh(m \log^2(n) + n \log^4(n))$ by \cref{thm:karger1}.  Step~2 takes time $\Oh(m \log^2 n)$ by \cref{lem:mlogn}.  In the for loop, Step~4 takes time $\Oh(m)$ by \cref{lem:one_respect}; Step~5 takes time $\Oh(m \log^4 n)$ by \cref{thm:sp-tree}; Step~6 takes time 
$\Oh(n \log^{2}n\log^{*}n)$ by \cref{lem:euler-tour}.  Thus the time in the for loop is dominated by Step~4, and the total time taken over the $K = \Oh(\log n)$ iterations is 
$\Oh(m \log^5 n)$.  The last step takes time $\Oh(n \log^2(n))$.  Thus the complexity overall is $\Oh(m \log^5 n)$.

To finish the proof of the theorem let us handle the case of $\bigwedge  \cB_\eps$.  We claim that given the value of $\lambda(G)$ we can compute $\bigwedge \cB_\eps$ from $\bigwedge \cB_\eps^{nt}$ deterministically 
in $\Oh(m)$ time.  In $\Oh(m)$ time we can identify 
the set $Z = \{v \in V : \Delta_G(\{v\}) \le (1+\eps) \lambda(G) \}$.  Let $\cD_\eps = \{ \{v, V \setminus v\} : v \in Z\}$ be the corresponding set of bipartitions and note that 
$\bigwedge  \cB_\eps = \left(\bigwedge \cB_\eps^{nt} \right) \wedge \left( \bigwedge \cD_\eps \right)$.  The meet $\bigwedge \cD_\eps$ is simply the partition consisting 
of the sets $\{v\}$ for $v \in Z$ and $V \setminus Z$.  To take the meet of this partition with $\cP = \bigwedge \cB_\eps^{nt}$ we simply cycle through each $S \in \cP$ 
and split $S$ into the sets $\{v\}$ for $v \in S \cap Z$ and $S \setminus Z$, which can be done in time $\Oh(n)$.  Thus the total time of computing  $\bigwedge  \cB_\eps$ is 
dominated by the computation of $\bigwedge \cB_\eps^{nt}$, and can be done asymptotically in the same time.
\end{proof}

\section{Applications}
In this section we give two applications of our main result: an improved quantum algorithm for minimum cut in weighted graphs in the adjacency list model, and a new 
randomized algorithm with running time $\Oh(m + n \log^6 n)$ to compute the edge connectivity of a simple graph.

\subsection{Quantum algorithm for minimum cut in weighted graphs} \label{sec:quant-algo}
In a recent work by Apers and Lee \cite{ApersL20} the quantum complexity of the minimum cut problem was studied.
They distinguish two models for querying a weighted graph as an input.
In the \emph{adjacency matrix} model a query is a pair of vertices $i,j \in V$ and the answer to the query reveals whether $\{i,j\} \in E$, and if so, also returns the weight $w(\{i,j\})$.
In the \emph{adjacency list} model a query is a vertex $i \in V$ and an integer $k \in [n]$, and the answer to the query is the $k$-th neighbor $j$ of vertex $i$ (if it exists) and the corresponding weight $w(\{i,j\})$.
The main results from \cite{ApersL20} depend on the edge-weight ratio $\tau$, defined as the ratio of the maximum edge weight over the minimum edge weight.
These results can be summarized as follows:
\begin{itemize}
\item
In the \emph{adjacency matrix model}, finding a minimum cut of a weighted graph with edge-weight ratio $\tau$ has quantum query and time complexity $\widetilde\Theta(n^{3/2} \sqrt{\tau})$.
This compares to the $\Theta(n^2)$ query complexity of any classical algorithm for minimum cut in this model \cite{DHHM06}.
\item
In the \emph{adjacency list model}, finding a minimum cut of a weighted graph with edge-weight ratio $\tau$ requires quantum query complexity $\widetilde\Oh(\sqrt{m n \tau})$ and quantum time complexity $\widetilde\Oh(\sqrt{m n \tau} + n^{3/2})$.
There are also lower bounds of $\Omega(n^{3/2})$ for $\tau > 1$ and $\Omega(\tau n)$ for $1 \leq \tau \leq n$.
This compares to the $\Theta(m)$ query complexity of any classical algorithm for minimum cut in this model \cite{BGMP20}.
\end{itemize}
While this fully resolves the quantum complexity of minimum cut in the adjacency matrix model, there are two apparent gaps in the adjacency list model.
On the one hand there is a gap between the upper and lower bounds on the quantum query complexity.
On the other hand there is a gap between the upper bounds on the quantum \emph{query} complexity and the quantum \emph{time} complexity.
Using our new result (\cref{thm:main}) we can close this second gap.

Let $\kappa(n)$ denote the (quantum) time complexity for finding a $(1+\eps)$-KT partition of a weighted graph with $n$ vertices and $\tOh(n)$ edges.
The following lemma is proven in \cite{ApersL20}.
\begin{lemma}[{\cite[Lemma 22]{ApersL20}}]
Let $G$ be a weighted graph with $n$ vertices, $m$ edges, and edge-weight ratio $\tau$.
There is a quantum algorithm to compute the weight and shores of a minimum cut of $G$ with time complexity $\kappa(n) + \tOh(\sqrt{mn\tau})$ in the adjacency list model.
\end{lemma}

In \cite{ApersL20} a \emph{quantum} algorithm was proposed for finding the KT partition of a weighted graph with $m$ edges in time $\tOh(m+n^{3/2})$, giving an upper bound $\kappa(n) \in \tOh(n^{3/2})$ and an upper bound $\tOh(\sqrt{mn\tau} + n^{3/2})$ on the quantum time complexity.
Our main result gives a \emph{classical} algorithm that improves this upper bound to $\kappa(n) \in \tOh(n)$, and hence this yields a quantum algorithm for minimum cut with time complexity $\tOh(\sqrt{mn\tau})$.

\begin{corollary}
Let $G$ be a weighted graph with $n$ vertices, $m$ edges, and edge-weight ratio $\tau$.
There is a quantum algorithm to compute the weight and shores of a minimum cut of $G$ with time complexity $\tOh(\sqrt{mn\tau})$ in the adjacency list model.
\end{corollary}

\subsection{Randomized algorithm for edge connectivity} \label{sec:conn-algo}
We can use our algorithm for finding the KT partition of a \emph{weighted} graph to give a randomized algorithm that computes the edge connectivity of a \emph{simple} graph $G$ with high probability in time 
$\Oh(m + n \log^{6} n)$.  For graphs that are not too sparse this equals the best known $\Oh(m + n \log^2 n)$ complexity of the random contraction based algorithm by Ghaffari, Nowicki and Thorup \cite{GNT20}.

Our new algorithm uses the key idea from Kawarabayashi and Thorup \cite{KT19}: (i) find the KT partition of the graph and contract the components of the partition, and (ii) find a minimum cut in the contracted graph.
By definition of the KT partition, this contraction will preserve the set of non-trivial minimum cuts, and so it suffices to find a minimum cut in the contracted graph and the minimum degree 
of a vertex.  Moreover, the contracted graph has only $\Oh(n)$ edges and so we can find a minimum cut in this graph quickly.

Our algorithm follows the same blueprint, except that in order to obtain $\Oh(m)$ leading complexity we first find an \emph{$\eps$-cut sparsifier} $F$ of the input graph, for a small constant $\eps$.
For this step we can use the $\Oh(m)$ sparsification algorithm from Fung, Hariharan, Harvey and Panigrahi \cite[Theorem 1.22]{FHHP19}.  Provided that the sparsification step is successful, any
minimum cut of the original simple graph $G$ will be a $\gamma$-near minimum cut of $F$ for $\gamma = (1+\eps)/(1-\eps) \le 1+ 3\eps$.  Thus if we find a $(1+3\eps)$-KT partition of $F$ and 
contract the sets of resulting partition in $G$ we obtain a multigraph $G'$ which preserves all non-trivial minimum cuts of $G$.  
In this way we only need to find the KT partition of $F$, which has $\Oh(n \log n)$ edges rather than $m$ edges.  On the other hand, 
the sparsifier $F$ will in general be weighted, and hence we cannot run the near-linear time algorithm from \cite{KT19} to find its KT partition.  This is a prime example where 
finding the KT partition of a weighted graph is very useful.  

The next theorem fleshes out this algorithm.  For this, we need the fact that for a simple graph there are only $\Oh(n)$ inter-component edges in a KT partition.  We take the 
version from \cite{ApersL20} which gives an explicit constant in the bound.
\begin{lemma}[{\cite[Lemma 2.6]{RSW18},\cite[Lemma 2]{ApersL20}}]
\label{lem:RSW}
Let $G = (V,E)$ be a simple graph with $|V| =n$.  Let $d = \min_{u \in V} \deg(u)$.  For a nonnegative $\varepsilon < 1$, let 
$\cT = \{X : |X|, |\overline{X}| \ge 2 \mbox{ and } |\Delta_G(X)| \le \lambda(G) + \varepsilon d\}$ and let $G'$ be the multigraph 
formed from $G$ by contracting the sets in $\bigwedge \cT$.  Then 
\[
|E(G')| \le \frac{68 n}{(1-\varepsilon)^2} \enspace .
\]
\end{lemma}

\begin{algorithm}[H]
\caption{Randomized algorithm for edge connectivity}
\label{alg:classicalmincut}
 \hspace*{\algorithmicindent} \textbf{Input:} Adjacency list access to a simple graph $G$. \\
 \hspace*{\algorithmicindent} \textbf{Output:} A minimum cut of $G$.
\begin{algorithmic}[1]
\State
Find a vertex $v$ with minimum degree $d_{\min}$.  \Comment{$\Oh(m)$ time.}
\State
Construct an $1/100$-cut sparsifier $F$ of $G$ with $\Oh(n \log n)$ edges. \Comment{$\Oh(m)$ time by \cite{FHHP19}.}
\State
Find the $(101/99)$-KT partition $\mathcal{S} = \{S_1,\dots,S_k\}$ of $F$ using \cref{thm:main}. \Comment{$\Oh(n \log^6 n)$ time.}
\State
Contract the components $S_1,\dots,S_k$ and let $G'$ be the resulting multigraph.  If $G'$ has at most $100n$ edges find a minimum cut $C$ of $G'$, otherwise abort. 
\Comment{Time $\Oh(m+n \log^{2} n)$ using the minimum cut algorithm of \cite{GawrychowskiMW20} from \cref{lem:mlogn}.} 
\State
If $d_{\min} \leq |C|$, return the outgoing edges from $v$.
Otherwise, return $C$.
\end{algorithmic}
\end{algorithm}

\begin{theorem}
Let $G$ be a simple graph with $m$ edges.
There is a classical randomized algorithm that runs in time $\Oh(m + n \log^{6} n)$ and with high probability outputs the edge connectivity of $G$ and a cut realizing this value.
\end{theorem}

\begin{proof}
The algorithm is given in \cref{alg:classicalmincut}.  The time complexity of each step is given in the comments.  Let us prove correctness.  

The algorithm either outputs a trivial cut or a cut from a contraction $G'$ of $G$.  As contraction cannot decrease the edge connectivity, if the edge connectivity 
of $G$ is realized by a trivial cut the algorithm will be correct.  Let us now assume that the edge connectivity $\lambda(G)$ is realized by a non-trivial cut $C^* = \Delta_G(S)$.  
In step~2 we use the sparsification algorithm of Fung, Hariharan, Harvey and Panigrahi \cite[Theorem 1.22]{FHHP19} to find a $1/100$-cut sparsifier $F = (V, E_F, w_F)$ of $G$, 
which succeeds with high probability.  Thus with high probability $w_F(\Delta_F(S)) \le (1+1/100) \lambda(G)$.  Also with high probability the weight of a minimum cut of $F$ is at least
$(1-1/100)\lambda(G)$, in which case $\Delta_F(S)$ will be a $101/99$-near minimum cut of $F$.  Hence with high probability the $(101/99)$-KT partition of $F$ will be a 
refinement of $\{S, \bar{S}\}$, and in the contraction $G'$ it will hold that $|\Delta_{G'}(S)| = \lambda(G)$ and so the edge connectivity of $G'$ is $\lambda(G)$.  
Further, if $F$ is a valid $1/100$-cut sparsifer of $G$ it will hold that 
$G'$ has at most $100n$ edges by \cref{lem:RSW} and so we can find a minimum cut $C$ of $G'$ in time $\Oh(n \log^2 n)$ using the minimum cut algorithm of \cite{GawrychowskiMW20} 
given in \cref{lem:mlogn}.
Thus in this case with high probability $C$ will be a cut realizing the edge connectivity of $G$ and the algorithm is correct.
\end{proof}

\section{Discussion}
We find the $(1+\eps)$-KT partition of a weighted graph in near-linear time for any $0 \leq \eps \leq 1/16$.  The near-linear time deterministic algorithm of Kawarabayashi and Thorup \cite{KT19}
to find a KT-partition of a simple graph differs from ours with respect to the parameters in an interesting respect.  Recall that we defined $\cB_\eps^{nt}(G)$ to be the set of all bipartitions $\{S, \bar{S}\}$ of the 
vertex set corresponding to non-trivial cuts whose weight is at most $(1+\eps)\lambda(G)$, and a $(1+\eps)$-KT partition to be $\bigwedge \cB_\eps^{nt}$.  Kawarabayashi and Thorup 
consider the larger set of bipartitions $\cK_\eps^{nt}(G)$ corresponding to non-trival cuts of weight at most $\lambda(G) + \eps d$, where $d$ is the minimum degree of $G$.  When $G$ is simple
they can compute $\bigwedge \cK_\eps^{nt}(G)$ for any $\eps < 1$ in near-linear time.  Thus it is stronger than our result with respect to the parameters in two ways: it allows any $\eps < 1$ and also lets 
$\eps$ multiply the minimum degree rather than $\lambda(G)$.  

There is an inherent barrier to extending the 2-respecting cut framework we employ here to this parameter regime.  The reason is that Karger's tree packing lemma \cite[Lemma 2.3]{Karger} only shows that a cut of weight 
$< 3 \lambda(G)/2$ will 2-respect a positive fraction of the trees from a maximum tree packing.  To handle cuts of weight $3 \lambda(G)/2$ one would have to move instead to considering 3-respecting cuts, which seems to 
add a good deal of complexity.  Thus while we have not tried to optimize the constant $1/16$, there is a natural barrier to extending our methods for $\eps \ge 1/2$.  
Pushing to larger $\eps$ and also allowing $\eps$ to multiply the minimum weight of a vertex rather than $\lambda(G)$ both seem to require new techniques, and we leave this as an open question.

\bibliographystyle{alpha}
\bibliography{biblio}

\appendix

\section{Data structures} \label{app:data-structure}

We first show how to implement categorical top two queries on an array while allowing updates to add $\Delta$ to the scores in an interval.  This 
can be accomplished using a well-known binary tree data structure.  We will then port this construction to a tree $T$ by means of the heavy path 
decomposition of $T$ \cite{ST1983,HT1984}.  

The key to the binary tree data structure is the following simple fact.  For a node $u$ of a tree let $\interval(u)$ be the set of labels of leaves that are descendants of $u$.  
\begin{fact}
\label{fact:int_decomp}
Let $n$ be a power of 2 and $T$ a complete binary tree with $n$ leaves labeled by $1, \ldots, n$.  
For any interval $[i,j]$ there are $\Oh(\log n)$ many nodes $u_1, \ldots, u_t$ such that $[i,j] = \interval(u_1) \sqcup \cdots \sqcup \interval(u_t)$.  Moreover $u_1, \ldots, u_t$ 
can be found in $\Oh(\log n)$ time, and the total number of ancestors of $u_1, \ldots, u_t$ is $\Oh(\log n)$.
\end{fact}

\toptwopath*

\begin{proof}
By padding the array with scores of infinity and an arbitrary color we may assume that $n$ is a power of $2$. 
The data structure will be a complete binary tree $B$ with $n$ leaves labeled as $\ell_1, \ldots, \ell_n$.  
Each leaf stores a three tuple $(\score, \mathrm{index}, \edgecolor)$ and at leaf $i$ this three tuple is initialized to 
be $(A[i].\score,i,A[i].\edgecolor)$.  Every internal node $u$ will store a pair of such 3-tuples.  The data structure will 
maintain the invariant (Invariant 1) that at every internal node $u$ the indices in this pair of three tuples is the answer to 
the categorical top two query for the interval $\interval(u)$.  The answer to a categorical top two query 
for the interval $\interval(u)$ can be computed in constant time from answers to this query at the children of $u$.  Thus in time $\Oh(n)$ we can propagate 
the answers to the categorical top two queries from the leaves to the root so that Invariant 1 holds.

Each node $u$ will also store an $\update$ value $u.\update$.  We initialize the leaves to have $\ell_i.\update = A[i].\score$ 
and set the update value of all internal nodes of the tree to be zero.  Thus we have the property (Invariant 2) that the sum of the update 
values from $\ell_i$ to the root is $A[i]$, which will be maintained under the $\textsc{Add}(\Delta,i,j)$ updates.  
This completes the pre-processing step and the total pre-processing time is $\Oh(n)$.  

We now show that after an update we can adjust the tree to maintain Invariant 1 and Invariant 2
in $\Oh(\log n)$ time.  If the invariants hold, then we can answer a categorical top two query for the interval $[i,j]$ in time $\Oh(\log n)$.  This is done by first using 
\cref{fact:int_decomp} to find in $\Oh(\log n)$ time nodes $u_1, \ldots, u_t$ such that $\interval(u_1), \ldots, \interval(u_t)$ form a partition of $[i,j]$.  Then by building a binary tree on top of 
$u_1, \ldots, u_t$ and propagating the categorical top two query answers up this tree we can answer the categorical top two query for $[i,j]$ in time 
$\Oh(\log n)$.  

To restore the invariants after $\textsc{Add}(\Delta,i,j)$, we use \cref{fact:int_decomp} to find in $\Oh(\log n)$ time
nodes $u_1, \ldots, u_t$ such that $\interval(u_t), \ldots, \interval(u_t)$ 
form a partition of $[i,j]$. Then for each $i=1, \ldots, t$ we set $u_i.\update \gets u_i.\update + \Delta$.  This restores Invariant 2 under the update.  To restore Invariant 1, we recompute the answers 
to the categorical top two query at all ancestors of $u_1, \ldots, u_t$.  By \cref{fact:int_decomp} there are only $\Oh(\log n)$ many such ancestors, thus we can perform this computation in $\Oh(\log n)$ time as well.
\end{proof}

In order to extend this structure to a general tree $T$, we first construct its heavy path decomposition.
Next, we concatenate the heavy paths to form a list of all edges of $T$ with the property that
any subtree $T_e$ is described by a contiguous range of edges (but potentially containing many heavy paths).
This is done recursively as follows.
Let the topmost heavy path be $h=u_{1}-u_{2}-\ldots u_{k}$. We first write down its edges $(u_{1},u_{2}),(u_{2},u_{3}),\ldots,(u_{k-1},u_{k})$.
Then, we remove them from the tree.
We recurse on the trees consisting of more than one node rooted at $u_{k},u_{k-1},\ldots,u_{1}$ (note that $u_{k}$ is always a root
of tree consisting of size 1), in this order. This guarantees that, for any $e\in h$, $T_{e}$ indeed consists of a
contiguous range of edges, while for other edges this is guaranteed recursively.

\toptwotree*

\begin{proof}
Consider a heavy path decomposition of $T$, and construct the edge array $A[1..(n-1)]$ by concatenating the heavy paths
as described above.
We will use the data structure from \cref{lem:top2path} on $A[1..(n-1)]$.
Any path $p$ can be decomposed into $\Oh(\log n)$ infixes of heavy paths (in fact, at most one proper infix and
a number of prefixes), and hence it corresponds to $\Oh(\log n)$ contiguous ranges of $A[1..(n-1)]$.
Hence we implement the first operation by making $\Oh(\log n)$ calls to $\textsc{Add}(\Delta,i,j)$,
by Lemma~\ref{lem:top2path} this takes $\Oh(\log^{} n)$ time.
Finally, since $T_e$ is described by a single contiguous range $A[i..j]$, 
a categorical top-two query in $T_{e}$ corresponds to operation $\textsc{CatTopTwo}(i,j)$, which again takes time $\Oh(\log n)$.
\end{proof}

\end{document}